\documentclass[10pt,english]{article}
\usepackage[utf8]{inputenc}
\usepackage{cite}
\usepackage[shortlabels]{enumitem}
\usepackage{amsmath,amsthm,amssymb}
\usepackage{longtable}
\usepackage{array}
\usepackage{color}
\usepackage{xspace}
\usepackage{multirow}
\usepackage[framemethod=TikZ]{mdframed}
\usetikzlibrary{shapes.symbols} 
\usetikzlibrary{arrows.meta}
\usetikzlibrary{calc}
\usepackage{float}
\usepackage[]{graphicx}
\usepackage{alltt}
\usepackage{diagbox}
\usepackage[T1]{fontenc}
\usepackage[utf8]{inputenc}
\usepackage[top=100pt,bottom=100pt,left=68pt,right=66pt]{geometry}
\usepackage{lipsum}
\usepackage[english]{babel}
\usepackage{fancyhdr}
\usepackage{titlesec}
\titleformat{\chapter}
   {\normalfont\LARGE\bfseries}{\thechapter.}{1em}{}
\titlespacing{\chapter}{0pt}{50pt}{2\baselineskip}
\usepackage{float}
\usepackage{soul}
\floatstyle{plaintop}
\restylefloat{table}
\usepackage{hyperref}
\hypersetup{hidelinks,linkcolor = black} 
\usepackage{doi}    
\graphicspath{ {images/} }

\newtheorem{defi}{Definition}[section]
\newcommand{\definition}[3]{\begin{mdframed}[roundcorner=6pt,linecolor=blue,backgroundcolor=blue!5]
\begin{defi}[#1]
\label{#2}
#3
\end{defi}
\end{mdframed}
}

\newtheorem{thm}{Theorem}[section]
\newenvironment{theorem}{\begin{mdframed}[roundcorner=6pt,linecolor=red,backgroundcolor=red!5]
\begin{thm}}{\end{thm}
\end{mdframed}}

\newtheorem{cor}{Corollary}[section]
\newenvironment{corollary}{\begin{mdframed}[roundcorner=6pt,linecolor=green,backgroundcolor=green!5]
\begin{cor}}{\end{cor}\end{mdframed}}

\newtheorem{mdl}{Model}[section]
\newenvironment{model}{\begin{mdframed}[roundcorner=6pt,linecolor=black,backgroundcolor=black!5]
\begin{mdl}}{\end{mdl}
\end{mdframed}}

\newtheorem{exampl}{Example}[section]
\newcommand{\example}[3]{\begin{mdframed}[roundcorner=6pt,linecolor=black,backgroundcolor=black!5]
\begin{exampl}[#1]
\label{#2}
#3
\end{exampl}
\end{mdframed}
}

\newcommand{\personagent}{\textit{person-agent}\xspace}
\newcommand{\pa}{\textit{pa}\xspace}
\newcommand{\interfaceagent}{\textit{interface-agent}\xspace}
\newcommand{\Personagent}{\textit{Person-agent}\xspace}

\newcommand{\Interfaceagent}{\textit{Interface-agent}\xspace}

\newcommand{\location}{\textit{location}\xspace}
\newcommand{\locationcollection}{\textit{location collection}\xspace}

\newcommand{\regionallevel}{\textit{regional-level}\xspace}
\newcommand{\regionfamily}{\textit{region-family}\xspace}
\newcommand{\regionmapping}{\textit{region-mapping}\xspace}
\newcommand{\regionid}{\textit{region-id}\xspace}

\newcommand{\regionallevels}{\textit{regional-levels}\xspace}
\newcommand{\regionfamilies}{\textit{region-families}\xspace}

\newcommand{\regionids}{\textit{region-ids}\xspace}

\colorlet{mycolor}{green!30}
\sethlcolor{mycolor}
\newcommand{\important}[1]{
\hl{#1}
}

\renewcommand{\important}[1]{#1}

\begin{document}

\selectlanguage{english}

\begin{titlepage}
	\clearpage\thispagestyle{empty}
	\centering
	\includegraphics[width=0.2\linewidth]{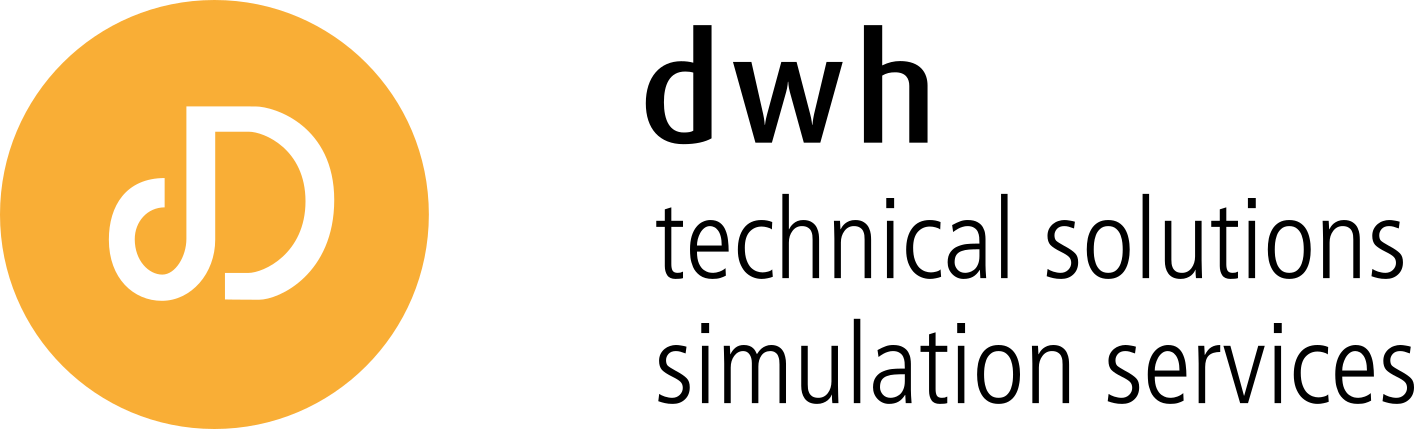}\hspace{1cm}
	\includegraphics[width=0.2\linewidth]{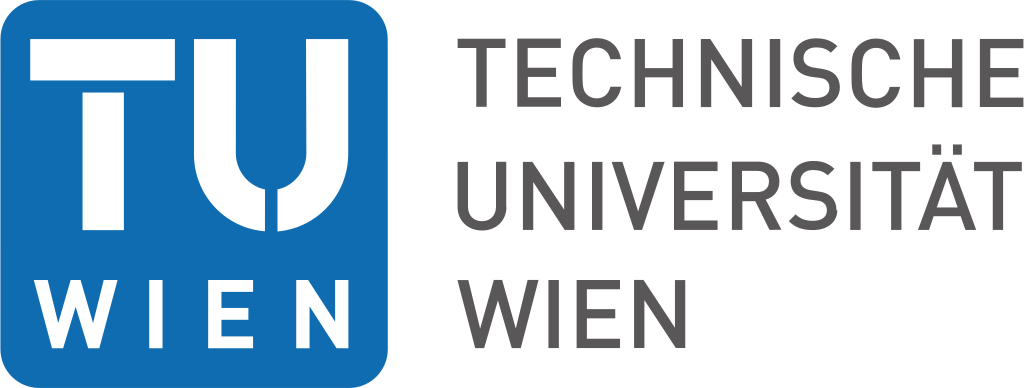}\\
    \vspace{0.5cm}
	\vspace{1cm}
	\rule{12cm}{1.2pt}\\
	\vspace{0.5cm}
	{\Huge \textbf{GEPOC ABM}} \\
 \vspace{0.1cm}
   {\Large \textbf{Generic Population Concept - Agent-Based Model}}\\
   {\Large \textbf{Version 2.2}}\\
	\vspace{0.5cm}
	{\large \textbf{Model Definition}} \\
	\vspace{1cm}
	{\large \today}\\
	\vspace{0.5cm}
    {\large
    Martin Bicher$^{1,2,*}$, Dominik Brunmeir$^{1,2}$, Claire Rippinger$^{1}$, Christoph Urach$^{1,3}$,
    Maximilian Viehauser$^{1,3}$, Daniele Giannandrea$^{1,3}$, Hannah Kastinger$^{1}$, Niki Popper$^{1,2,3}$\\
    }
    \vspace{0.5cm}
    {\small
    $^{1}$ dwh GmbH, dwh simulation services, Neustiftgasse 57--59, 1070 Vienna, Austria\\
    $^{2}$ TU Wien, Institute of Information Systems Engineering, Favoritenstra\ss e 9-11, 1040 Vienna, Austria\\
    $^{3}$ TU Wien, Institute of Statistics and Mathematical Methods in Economics, Wiedner Hauptstraße 8-10, 1040 Vienna, Austria\\
    $^{*}$Correspondence: martin.bicher@dwh.at; martin.bicher@tuwien.ac.at
    }\\[0.5em]
	\rule{12cm}{1.2pt}\\
    \vspace{2cm}
    {\normalsize
    \begin{abstract}
        The Generic Population Concept - Agent-Based Model, henceforth short, GEPOC ABM, is one of the models within GEPOC, a generic concept to model a country's population and its dynamics using causal modelling approaches. The model is well established and had already proven its worth in various use cases from evaluation of MMR vaccination rates to SARS-CoV-2 epidemics modelling. In this work we will reproducibly specify the base model, to be specific, version 2.2 of it, and several extensions. The base model GEPOC ABM depicts the population of a country with the features sex and age. It uses a co-simulation-inspired time-update, where person-level discrete-event simulators are synchronised by a simulation layer at macro-steps, making the approach amenable to parallelization. A core design choice is structuring person agents around the life-year rather than the calendar year; accordingly, each agent schedules demographic events annually on their birthday. To expand the model's capabilities beyond basic demographic features, GEPOC ABM Geography adds a residence feature in the form of geographical coordinates. Further extensions include GEPOC ABM IM, which adds internal migration processes in three variants, and GEPOC ABM CL, which models locations where agents may have contacts with each other. In this definition we solely specify the conceptual models and do not go into any details with respect to implementation or gathering/processing of parametrisation data.
    \end{abstract}
    }
	\pagebreak

\end{titlepage}
\tableofcontents
\newpage
\section{Introduction}
The Generic Population Concept - Agent-Based Model, henceforth short GEPOC ABM, is one of the models within GEPOC, a generic concept to model a country's population and its dynamics using causal modelling approaches. By 2025, the other models in the concept are a system-dynamics model GEPOC SD~\cite{bicher_definition_2015} and a partial differential equation model GEPOC PDE~\cite{bicher2018mean}. Goal of generic population concept GEPOC is to establish valid and flexible base models for population-focused research questions. By now, GEPOC ABM is by far the most successful of the three mentioned models with respect to applications.

In the following we will provide a model definition of the current conceptual model of GEPOC ABM and of three extensions of it. We hereby put specific emphasis on the term conceptual, since we do not specify how the model could be implemented, where data for parametrisation could be found and accessed, or how raw data could be processed for parametrisation. These challenges might be equally or even more difficult than the conceptualisation of the model itself. This documentation refers to Version 2.2 of GEPOC ABM and extends Version 1, published in \cite{bicher_definition_2015}, by all geographic features and Version 2.1, published in \cite{bicher_gepoc_2018}, by an improved global and individual time-update scheme and updated spatial population sampling mechanic.

\section{Definition/Communication Strategy}
\label{sec:communication_strategy}
The model and its extensions will be specified based on the \important{ODD (Overview, Design Concepts, Details) protocol} by Volker Grimm et.al. \cite{grimm_standard_2006,grimm_odd_2010}. This protocol mainly provides standardised headlines and defines in which order certain model parts are presented.

Since the model uses various continuous-time (i.e. discrete event) aspects, an additional visual concept is used to depict the underlying dynamics, a customised \important{event graph} notation. In their foundation, the diagrams use the syntax from \cite{schruben1983simulation} including parametrised events and cancelling edges\footnote{That means, scheduling edges may have (a) a time-delay, indicated by a variable or number directly above or below the start of the edge, (b) a condition, indicated by a $\int$ sign with condition text next to it, and/or (c) parameters, indicated by boxed variables. Event nodes may have (a) one or many parameters, indicated by round parentheses below the event name, and/or (b) a state effect, indicated by assignments and terms under the node. Cancelling edges are highlighted by dashed arrows.}. To customise the notation we introduced the following adaptions and conventions to the original syntax (see Figure \ref{fig:update_concept_sim} as example):
\begin{itemize}
\item Boxes indicate interfaces between the different layers. They can be regarded as parameterised events which are scheduled by an origin (``from'') into the event queue of a target (``to''). The event notice is hereby added to the queue of the recipient as if it was scheduled by its DES without any time-delay. That means, the original schedule time of the event in the origin layer is irrelevant and must be passed on as additional parameter, if needed. Colours are used to highlight the connections between the three diagrams.
\item We implicitly assume, that all layers reachable via the interfaces exist. We shift the problem of correctly instantiating and deleting agents to the model implementation.
\item State variable assignments are specified by $\leftarrow$, defining variables with $=$ indicates only local and temporal use as parameters within scheduling edges.
\item Functions $f_1,f_2,\dots$ indicate computations which are too comprehensive to be described within the event graph. They are explained in the main text and, in detail, described in Section \ref{sec:parameter_functions}.
\item Variables $p,b,e,i,d$ indicate demography-specific parameter functions which may depend on time, sex and age. They are explained in the main text and, in detail, in Section \ref{sec:parameter_functions}.
\item Variables $U_{i},i=1,\dots,$ define uniformly distributed random numbers between zero and one and are drawn independently at the time at which the event is executed.
\end{itemize}

\newpage
\section{GEPOC ABM Base - Model Definition}
\label{sec:basemodel}
\subsection{Overview}

\subsubsection{Purpose and Patterns}
GEPOC ABM serves the purpose of a base-model for research questions which rely on the population of a country or region and require a microscopic representation of the population. To fulfil this purpose, the model must be capable of (a) creating a valid microscopic image of the population at a given point in time and (b) validly depict the dynamics of this image for a given time-span. The validity is measured qualitatively and quantitatively.

\subsubsection{Entities, State Variables and Scales}
GEPOC ABM uses two types of agents, {\personagent}s, short {\pa}s and {\interfaceagent}s.

\important{\mbox{\Personagent}s represent the actual individuals of the country's population.} 
\important{Each {\textit{pa}} has two parameters which are set at initialisation:}
\begin{itemize}
\item \important{date of birth (\textit{birthdate}), and}
\item \important{sex at birth (\textit{sex}).}
\end{itemize}
Hereby, \textit{sex} is a binary variable and can either be \textit{male} or \textit{female} -- see below for a precise interpretation. Furthermore, an agent's \textit{age} is regarded as dependent state-variable of the agent and is computed from its date-of-birth and the current simulation time.

Usually, one \textit{pa} represents one natural individual in reality, yet, the model can be scaled by an arbitrary scaling factor $\sigma$ so that 
\[1\ \personagent\ \equiv \sigma\text{ real persons}.\]
In this situation, one agent in the model statistically represents $\sigma$ persons in reality\footnote{Usually $\sigma>1$ is chosen only if computation time is an issue and if the research question allows it.}.

In addition the model uses one \interfaceagent, in prior work often called government-agent. \important{The \mbox{\interfaceagent} is responsible for the interface between country/region and the rest of the world.}
In base GEPOC ABM its primary target is to sample and introduce immigrated agents into the model. It is not regarded to have a certain state\footnote{Even though this entity does not have a state on its own, we would still consider the \interfaceagent an agent, since it highly interacts with the {\pa} population}.

\paragraph{Sex.}
Considering the sensitivity of the topic, the agent variable \textit{sex} requires an accurate interpretation w.r. to what real-world element it depicts:
\definition{\textit{sex}}{def:sex}{Persons with female biological sex at birth are modelled by agents with \textit{sex}=\textit{female}. As in reality, they have the potential to create offspring. All other persons are modelled by agents with \textit{sex}=\textit{male}.}
Including gender or non-binary \textit{sex} is not included in the base model, since it is not relevant for demography dynamics.

\subsubsection{Process Overview and Scheduling}
\label{sec:basemodel_process}
\important{In general GEPOC ABM is updated using a hybrid time-update concept in between a classical time-discrete and a discrete-event (DE) approach. The overall simulation unit uses a tick-based update system, the agents update with separate DESs. Hence, the overall dynamics can be described in three layers:}
\begin{itemize}
    \item \important{Simulation layer,}
    \item \important{\mbox{\Personagent} layer, and}
    \item \important{\mbox{\Interfaceagent} layer.}
\end{itemize}
\important{
The overall model can be interpreted as a co-simulation, where the simulation layer advances time on discrete ticks and synchronizes all lower-level discrete-event simulators (DESs). These comprise the person agents, each with its own event queue, and an interface agent that generates immigration events. This hierarchical arrangement yields a highly scalable, modular architecture with short local event queues.} This concept is shown in Figure \ref{fig:advancement}.

\begin{figure}[H]
    \centering
    \begin{tikzpicture}[
arrow/.style={-{Latex[scale=1.1]}},
arrow2/.style={dotted, -{Stealth[harpoon]}}
]
\draw (-2,0) edge[-{Stealth[scale=1.5]}] node[pos=0.5,label={above:simulation time}] {} (12,0);
\node[circle,draw=blue,label={[blue]left:simulation-layer},label={above:$t_0$}] (t0) at (0,-1) {};
\node[circle,draw=blue,label={above:$t_1$}] (t1) at (5.5,-1) {};
\node[circle,draw=blue,label={above:$t_2$}] (t2) at (11,-1) {};

\draw (t0.east) edge[arrow,color=blue] (t1.west);
\draw (t1.east) edge[arrow,color=blue] (t2.west);
\draw (t2.east) edge[arrow,color=blue] ($(t2.east)+(0.7,0)$);

\node[label={left:DES of agent $1$}] (al) at (0,-2) {};
\node[circle,draw] (a1) at (1,-2) {};
\node[circle,draw] (a2) at (3,-2) {};
\node[circle,draw] (a3) at (7,-2) {};
\node[circle,draw] (a4) at (9,-2) {};
\draw (al.east) edge[arrow] (a1);
\draw (a1) edge[arrow] (a2);
\draw (a2) edge[arrow] (a3);
\draw (a3) edge[arrow] (a4);
\draw (a4) edge[arrow] ($(t2.east)+(0.7,-1)$);

\node[label={left:DES of removed agent $2$}] (bl) at (0,-3) {};
\node[circle,draw] (b1) at (2,-3) {};
\node[circle,draw,inner sep = 0, minimum width = 9pt] (b2) at (4,-3) {\large{$-$}};
\draw (bl) edge[arrow] (b1);
\draw (b1) edge[arrow] (b2);
\draw (b2) edge[arrow] ($(t1.east)+(0,-2)$);

\node[label={left:DES of agent $3$}] (cl) at (0,-4) {};
\node[circle,draw] (c1) at (4,-4) {};
\node[circle,draw] (c2) at (8,-4) {};
\node[circle,draw] (c3) at (10,-4) {};
\draw (cl.east) edge[arrow] (c1);
\draw (c1) edge[arrow] (c2);
\draw (c2) edge[arrow] (c3);
\draw (c3) edge[arrow] ($(t2.east)+(0.7,-3)$);

\node[label={left:DES of new agent $4$}] (dl) at (0,-5) {};
\node[circle,draw,inner sep = 0, minimum width = 9pt] (d1) at (4,-5) {\large{$+$}};
\node[circle,draw] (d2) at (9,-5) {};
\draw (d1) edge[arrow] (d2);
\draw (d2) edge[arrow] ($(t2.east)+(0.7,-4)$);

\node[label={left:DES of new agent $5$}] (el) at (0,-6) {};
\node[circle,draw,inner sep = 0, minimum width = 9pt] (e1) at (3,-6) {\large{$+$}};
\draw (e1) edge[arrow] ($(t2.east)+(0.7,-5)$);

\node[label={left:DES interface agent}] (fl) at (0,-7) {};
\node[circle,draw,inner sep = 0, minimum width = 9pt] (f1) at (3,-7) {};
\draw (fl) edge[arrow] (f1);
\draw (f1) edge[arrow] ($(t2.east)+(0.7,-6)$);


\draw[color=blue] (t0.east) edge[arrow] ($(t0.east)+(0,-1)$) edge[arrow] ($(t0.east)+(0,-2)$) edge[arrow] ($(t0.east)+(0,-3)$) edge[arrow] ($(t0.east)+(0,-6)$);

\draw[color=blue] (t1.east) edge[arrow] ($(t1.east)+(0,-1)$) edge[arrow] ($(t1.east)+(0,-2)$) edge[arrow] ($(t1.east)+(0,-4)$) edge[arrow] ($(t1.east)+(0,-5)$) edge[arrow] ($(t1.east)+(0,-6)$);


\draw[color=blue] ($(t1.west)+(0,-6)$) edge[arrow] (t1.west);

\draw[color=blue] ($(t2.west)+(0,-6)$) edge[arrow] (t2.west);


\draw[color=blue] (al.east) edge[arrow2,bend left = 40] (a1);
\draw[color=black] (a1) edge[arrow2,bend left = 30] (a2);
\draw[color=blue] ($(t1.east)+(0,-1)$) edge[arrow2,bend left = 30] (a4);
\draw[color=black] (a2) edge[arrow2,bend left = 25] (a3);

\draw[color=blue] (bl.east) edge[arrow2,bend left = 30] (b1);
\draw[color=black] (b1) edge[arrow2,bend left = 30] (b2);

\draw[color=blue] (cl.east) edge[arrow2,bend left = 30] (c1);
\draw[color=red] (c1) edge[dotted,bend left = 40] ($(t1.west)+(0,-3)$);
\draw[color=black] (c1) edge[arrow2,bend left = 30] (c2);
\draw[color=black] (c2) edge[arrow2,bend left = 30] (c3);

\draw[color=red]  ($(t1.west)+(0,-3)$) edge[arrow2, bend right = 30] (d1);
\draw[color=black] (d1) edge[arrow2, bend left = 30] (d2);

\draw[color=blue] ($(t0.east)+(0,-6)$) edge[arrow2,bend left = 30] (f1);

\draw[color=red]  ($(t1.west)+(0,-6)$) edge[arrow2, bend right = 30] (e1);
\draw[color=red] (f1) edge[dotted,bend left = 30] ($(t1.west)+(0,-6)$);

\node[fill=blue,draw, circle,inner sep = 0, minimum width = 3pt]  at ($(t1.east)+(0,-2)$) {};
\end{tikzpicture}
    \caption{Process Overview of GEPOC ABM - Solid black arrows show local time advancement within each agent's DES; dotted harpoons depict event scheduling. The simulation layer provides the runtime infrastructure that orchestrates the agents' DESs. In each macro step, the simulation layer (blue) first observes the current states of the DESs (upward arrows) and may intervene (downward arrows) by scheduling external events (dotted blue arrows) or by adding and removing (events marked with $+$ and $-$) agents. Red dotted harpoons indicate interactions between agents via the simulation layer, here shown for a birth and an immigration event.}
    \label{fig:advancement}
\end{figure}
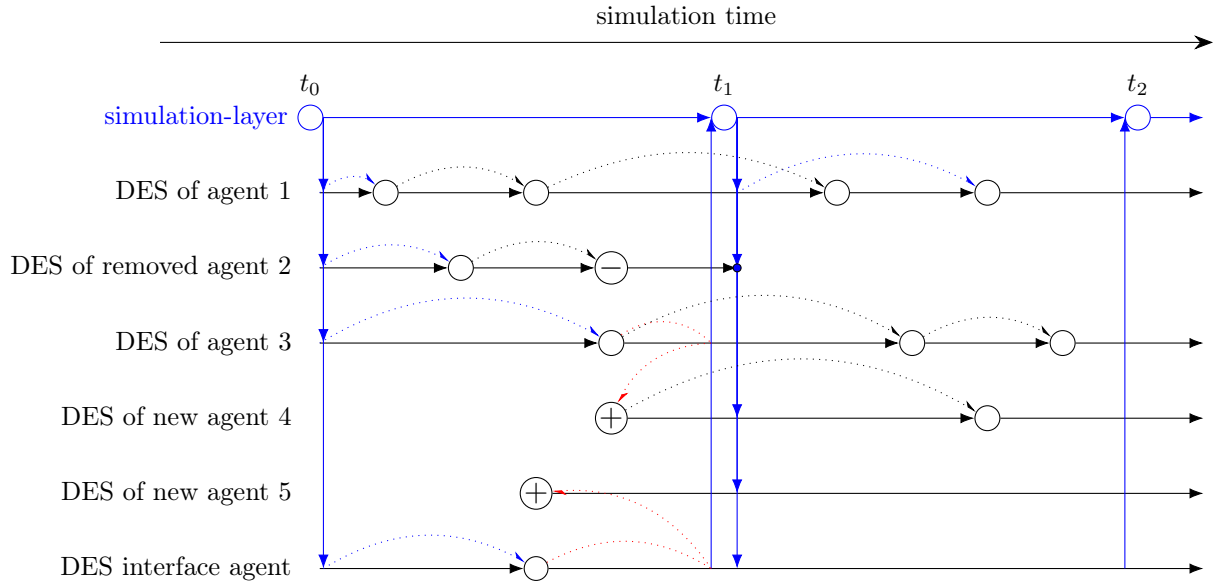

The three layers are furthermore defined using the corresponding event-graph diagrams (see Section \ref{sec:communication_strategy} for details on how to read these diagrams).
\newpage
\begin{figure}[H]
    \centering
    \begin{tikzpicture}[
eventtype/.style = {circle, draw, minimum width = 2 cm,text=black,draw=black, align = center},
interface/.style = {rectangle, draw, minimum width = 2 cm,text=black, align = center},
scheduleplain/.style = {-{Stealth[scale=1.5]}},
nodeif/.style = {pos=0.66,anchor=center, circle, minimum width=0.5cm},
nodetime/.style = {pos=0.05,anchor=center},
nodearg/.style = {pos=0.4,anchor=center, rectangle, minimum width=0.5 cm,fill=white,draw},
]
\node[eventtype,label={[align=left]below:$t\leftarrow 0$\\$N\leftarrow 0$}] (run) at (0,1) {Run};
\node[eventtype,label={[align=left]below:$(bd,s)=f_1(0)$\\$N++$}] (loop1) at (0,-3) {Init Loop\\ $(j)$};
\draw (loop1) edge[scheduleplain,out=90,in=0,looseness=3] node[nodearg]{$j+1$} node[pos=0.6,label={[align=left]right:$j<[p(0)/\sigma]$}] {} node[pos=0.6,rotate=140] {\large{$\int$}} (loop1);
\draw (run) edge[scheduleplain,bend right = 35] node[nodearg,pos=0.6] {$0$} (loop1);

\node[eventtype,font=\footnotesize] (tsp) at (5.5,0) {Time-Step\\Planning\\ $(i)$};
\node[eventtype] (loop2) at (6,-3) {Planning\\ Loop\\ $(j,i)$};
\draw (loop2) edge[scheduleplain,out=90,in=0,looseness=3] node[nodearg]{$j+1$} node[pos=0.6,label={[align=left]right:$j<N$}] {} node[pos=0.6,rotate=140] {\large{$\int$}} (loop2);
\draw (tsp) edge[scheduleplain,bend right = 20] node[nodearg] {$0,i$} (loop2);

\node[eventtype,font=\footnotesize] (tse) at (11,0) {Time-Step\\Execute\\ $(i)$};
\node[eventtype] (loop3) at (12,-3) {Execute\\ Loop\\ $(j,i)$};
\draw (loop3) edge[scheduleplain,out=90,in=0,looseness=3] node[nodearg]{$j+1$} node[pos=0.6,label={[align=left]right:$j<N$}] {} node[pos=0.6,rotate=140] {\large{$\int$}} (loop3);
\draw (tse) edge[scheduleplain,bend right = 30] node[nodearg] {$0,i$} (loop3);

\node[eventtype,font=\footnotesize] (obs) at (14,0) {Observe};
\draw (tse) edge[scheduleplain](obs);

\draw (run) edge[scheduleplain,bend right = 10] node[nodearg]{$0$} (tsp);
\draw (tsp) edge[scheduleplain,bend right = 20] node[nodetime,pos=0.18,label = above:$\Delta t_i$] {} node[nodearg]{$i$} (tse);
\draw (tse) edge[scheduleplain,bend right = 20] node[nodearg] {$i+1$} node[nodeif,label = above:$i<m$]{\large{$\int$}} (tsp);

\node[interface,draw=red,fill=red!10] (ini1) at (2,-7) {Interface $A_2$\\ to\\$\pa(j)$\\$(t,bd,s)$};
\draw (loop1) edge[scheduleplain,bend left = 20] node[nodearg,pos=0.7] {$j,0,bd,s$} (ini1);

\node[interface,draw=blue,fill=blue!10] (ini2) at (-0.5,-7) {Interface $A_1$\\ to\\\footnotesize{interface agent}\\$(t_0)$};
\draw (run) edge[scheduleplain,out=230,in=120] node[nodearg,pos=0.7] {$t_0$} (ini2);

\node[interface,draw=red,fill=red!10] (tsp1) at (8,-7) {Interface $B_2$\\ to\\$\pa(j)$\\$(t_i,t_{i+1})$};
\draw (loop2) edge[scheduleplain,bend left = 20] node[nodearg,pos=0.7] {$j,t_i,t_{i+1}$} node[nodeif,pos=0.17,label = {right:$\text{pa}(j):\textit{active}$}]{} node[pos=0.17,rotate=140] {\large{$\int$}} (tsp1);

\node[interface,draw=blue,fill=blue!10] (tsp2) at (5.5,-7) {Interface $B_1$\\ to\\\footnotesize{interface agent}\\$(t_i,t_{i+1})$};
\draw (tsp) edge[scheduleplain,out=230,in=120] node[nodearg,pos=0.7] {$t_i,t_{i+1}$} (tsp2);

\node[interface,draw=red,fill=red!10] (tse1) at (14,-7){Interface $C_2$\\ to\\$\pa(j)$\\$(t_{i+1})$};
\draw (loop3) edge[scheduleplain,bend left = 20] node[nodearg,pos=0.7] {$j,t_{i+1}$} node[nodeif,pos=0.17,label = {right:$\text{pa}(j):\textit{active}$}]{} node[pos=0.17,rotate=140] {\large{$\int$}} (tse1);

\node[interface,draw=blue,fill=blue!10] (tse2) at (11.5,-7)  {Interface $C_1$\\ to\\\footnotesize{interface agent}\\$(t_{i+1})$};
\draw (tse) edge[scheduleplain,out=230,in=120] node[nodearg,pos=0.7] {$t_{i+1}$} (tse2);

\node[eventtype,label={[align=left]below:$N++$\\$j=N$}] (add) at (2,-10.5) {Add\\$(t',bd,s)$};
\draw (add) edge[scheduleplain] node[nodearg,pos=0.3] {$j,t',bd,s$} (ini1);

\node[interface,draw=yellow,fill=yellow!10] (add1) at (7,-10.5) {Interface $D$\\ from\\agent\\$(t',bd,s)$};
\draw (add1) edge[scheduleplain] node[nodearg,pos=0.5] {$t',bd,s$} (add);

\draw (run) edge[scheduleplain,bend left = 20] (obs);
\end{tikzpicture}
    \caption{Uppermost layer of the time-update concept of GEPOC ABM using event-graph-like notation. Coloured boxes are the interfaces with the diagrams shown in Figure \ref{fig:update_concept_pa} and \ref{fig:update_concept_gov}. Functions $f_i$ are explained in the text.}
    \label{fig:update_concept_sim}
\end{figure}
\paragraph{Simulation Layer.} Figure \ref{fig:update_concept_sim} displays the event structure of the simulation-layer. Hereby, a discrete time-tick $t_i$ with $m$, not-necessarily equidistant, time-steps with lengths $\Delta t_i,i\in \{1,\dots,m\}$ [seconds] lies underneath. We write
\[t_i:=t_0+\sum_{j=1}^{i}\Delta t_j.\]

The diagram is explained in temporal order using the correct event-priorities.
\begin{itemize}
\item $t=t_0:$ \textit{Run}

\important{The \textit{Run} event triggers the start of the simulation. It sets the simulation time $t$ to $t_0$ which, is associated with the manually defined start-date $y_0\text{-}m_0\text{-}d_0TH_0\text{:}M_0\text{:}s_0$} See \ref{sec:parameter_functions} for details in this regard.

Moreover, the \textit{Run} event directly triggers \textit{Interface $A_1$} in the \interfaceagent (blue) and the \textit{Init Loop} event. The latter, alike all ``loop'' events, essentially iterates over the \pa population. In this case it creates the \pa indices $j$, updates the total number of agents $N$, and triggers \textit{Interface $A_2$} for all {\pa}s.\footnote{Specification of ``loops'' inside event graphs is, though technically correct, rather a misuse of the concept. Nevertheless, it often cannot be avoided if event graphs are used to describe ABMs.}

\important{Since \textit{Interfaces $A_1$} and $A_2$ trigger the \textit{Init} event in the corresponding agents, the ultimate goal of the \textit{Run} event is to properly initialise the agent population.}

As indicated by the rescheduling condition of the \textit{Init Loop} event, \textit{Interface $A_2$} is triggered $[p(t)/\sigma]$ times. Hereby $p(t)$ stands for the total population at time $t$, as given by a model parameter function (see below), $\sigma$ denotes the mentioned scaling factor of the model, and $[\cdot]$ indicates to round the number to the nearest integer. In every loop iteration a random birthdate $bd$ and sex $s$ is sampled which, in the diagram, is denoted by function $f_1$. For details regards $p$ and $f_1$, we refer to Section \ref{sec:parameter_functions}.

\item $t=t_0:$ \textit{Observe}

The \textit{Observe} event is used to \important{track the state of the simulation}. That means, the effect of this event can be defined in the specific implementation and may vary depending of the model-usage. Usually, aggregated numbers are collected by looping over the active agents. Most importantly the event \important{must not have any influence on the dynamics of the simulation itself} and its priority lies between the various \textit{Init} events and the \textit{Time-Step Planning} event.

\item $t=t_0:$ \textit{Time Step Planning}

\important{The \textit{Time Step Planning} event is used to schedule all simulation- and agent-specific actions for the upcoming time-tick $i$ from $t_i$ to $t_{i+1}$.} It essentially triggers all corresponding \textit{Time Step Planning} events of all agents via \textit{Interfaces B1} and \textit{B2}.

In contrast to the \textit{Init Loop} event, the \textit{Planning} and \textit{Execute Loop} event iterate over all {\pa}s, ever initialised via \textit{Interface A2}, but triggers the corresponding \textit{Interface B2} only for those, who are rendered \textit{active}. As seen in Figure \ref{fig:update_concept_pa} a \pa is initialised with $active=true$, but may be rendered inactive due to emigration or death.
d
\item $t=t_1:$ \textit{Time-Step Execute}.

\important{Between \textit{Planning} and \textit{Execute}, simulation time is advanced. While the main purpose of the prior was to schedule new events, the role of the latter is to advance event execution in the individual DESs of the agent layers, whereas the interface agent is prioritised. Note that all new agents created in the course of this time-step are already considered update (see below). This is done by triggering their \textit{Time-Step Execute} event via \textit{Interfaces C1} and \textit{C2}.}

\item $t=t_1:$ \textit{Add} ($\geq 0$ times)
Several \textit{Add} events are scheduled not from within but by the simulators of the individual agents via \textit{Interface D}. This is done as a consequence of birth and immigration events with the goal to increase the \pa population accordingly. With the introduced logic of the interfaces, the event notices arrive in the event queue of the simulation at the same time as the \textit{Execute Loop}, \textit{Observe} and \textit{Time-Step Planning} event. Compared to these it is specified to have a higher priority.
The precise creation time $t_0< t'\leq t_{1}$ of a newly created agent is passed as a parameter to \textit{Interface A2} which will also be treated as the initialisation time of the agent's DES. Since, the corresponding agent is set to \textit{active} and $N$ is increased, it will already be considered in the \textit{Execute Loop} of the current time-step to synchronise its DES to time-step $t_1$.

\item $t=t_1:$ \textit{Observe}

\item $t=t_1:$ \textit{Time-Step Planning} (+ Loop)

\item $t=t_2:$ \textit{Time-Step Execute} (+ Loop)

\item $t=t_2:$ \textit{Add} ($\geq 0$ times)

\item $t=t_2:$ \textit{Observe}

\item \dots

\end{itemize}
The loop breaks and the entire simulation stops as soon as the condition to schedule a new \textit{Time-Step Planning} event, $i<m$, is not met anymore.
\newpage
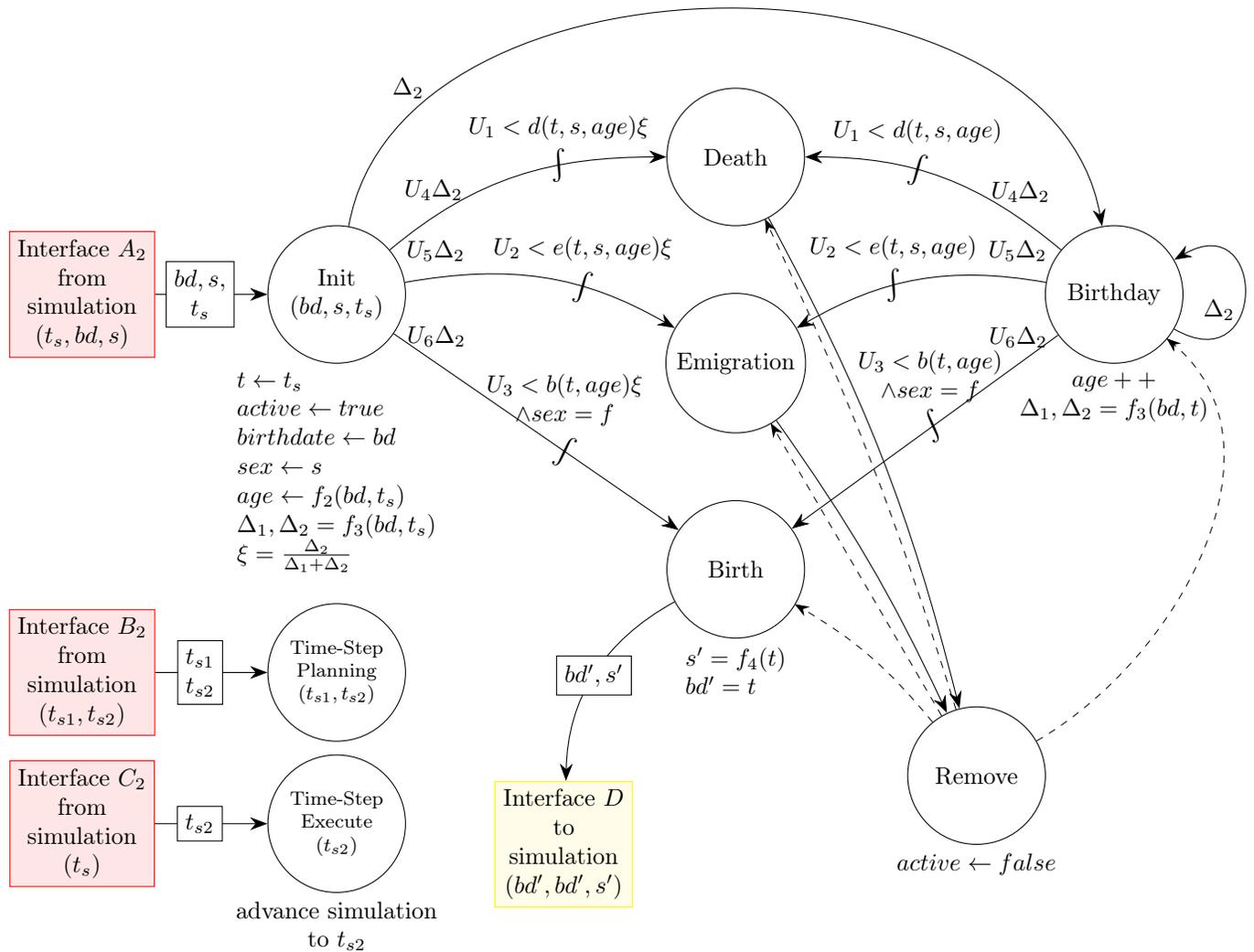
\begin{figure}[H]
    \centering
    \begin{tikzpicture}[
eventtype/.style = {circle, draw, minimum width = 2 cm,text=black,draw=black, align = center},
interface/.style = {rectangle, draw, minimum width = 2 cm,text=black, align = center},
scheduleplain/.style = {-{Stealth[scale=1.5]}},
nodeif/.style = {pos=0.66,anchor=center, circle, minimum width=0.5cm},
nodetime/.style = {pos=0.05,anchor=center},
nodearg/.style = {pos=0.4,anchor=center, rectangle, minimum width=0.5 cm,fill=white,draw},
]

\node[interface,draw=red,fill=red!10] (int1) at (-2,0) {Interface $A_2$\\ from\\simulation\\$(t_s,bd,s)$};

\node[interface,draw=red,fill=red!10] (int2) at (-2,-5.5) {Interface $B_2$\\ from\\simulation\\$(t_{s1},t_{s2})$};

\node[interface,draw=red,fill=red!10] (int3) at (-2,-7.7) {Interface $C_2$\\ from\\simulation\\$(t_s)$};

\node[eventtype,label={[align=left]below:$t\leftarrow t_s$\\$active\leftarrow true$\\$birthdate\leftarrow bd$\\$sex\leftarrow s$\\$age\leftarrow f_2(bd,t_{s})$\\$\Delta_1,\Delta_2=f_3(bd,t_{s})$\\$\xi=\frac{\Delta_2}{\Delta_1+\Delta_2}$}] (init) at (1.7,0) {Init\\$(bd,s,t_{s})$};

\node[eventtype, font=\footnotesize] (tsp) at (1.7,-5.5) {Time-Step\\Planning\\$(t_{s1},t_{s2})$};

\node[eventtype, font=\footnotesize,label={[align=center]below:advance simulation\\ to $t_{s2}$}] (tse) at (1.7,-7.7) {Time-Step\\Execute\\$(t_{s2})$};

\draw (int1) edge[scheduleplain] node[nodearg,pos=0.4,align=center]{$bd,s,$\\$t_s$} (init);

\draw (int2) edge[scheduleplain] node[nodearg,pos=0.4,align=center]{$t_{s1}$\\$t_{s2}$} (tsp);

\draw (int3) edge[scheduleplain] node[nodearg,pos=0.4,align=center]{$t_{s2}$} (tse);


\node[eventtype,label={[align=center]below:$age++$\\$\Delta_1,\Delta_2=f_3(bd,t)$}] (bd) at (13,0) {Birthday};
\draw (init) edge[scheduleplain,bend left = 80] node[nodetime,pos=0.15,label ={above:$\Delta_2$}] {} (bd);
\draw (bd) edge[scheduleplain,out = 330, in = 30, looseness = 4] node[nodetime,pos=0.2,label ={above:$\Delta_2$}]{} (bd);

\node[eventtype] (death) at (7.5,2) {Death};
\node[eventtype] (emi) at (7.5,-1) {Emigration};
\node[eventtype,label={[align=left] below:$s'=f_4(t)$\\$bd'=t$}] (birth) at (7.5,-4) {Birth};

\draw (init) edge[scheduleplain,out = 40,in = 180] node[nodetime,pos=0.17,label ={above:$U_4\Delta_2$}] {} node[nodeif,pos=0.65,label ={above:$U_1<d(t,s,age)\xi$}]{} node[pos=0.65,rotate=10] {\large{$\int$}} (death);

\draw (init) edge[scheduleplain,out = 10,in = 150] node[nodetime,pos=0.1,label ={above:$U_5\Delta_2$}] {} node[nodeif,pos=0.65,label ={above:$U_2<e(t,s,age)\xi$}]{} node[pos=0.65,rotate=170] {\large{$\int$}} (emi);

\draw (init) edge[scheduleplain] node[nodetime,pos=0.15,label ={above:$U_6\Delta_2$}] {} node[nodeif,pos=0.6,label ={[align=center]above:$U_3<b(t,age)\xi$\\$\wedge sex=f$}]{} node[pos=0.6,rotate=160] {\large{$\int$}} (birth);

\draw (bd) edge[scheduleplain,out = 140,in = 0] node[nodetime,pos=0.17,label ={above:$U_4\Delta_2$}] {} node[nodeif,pos=0.6,label ={above:$U_1<d(t,s,age)$}]{} node[pos=0.6,rotate=170] {\large{$\int$}} (death);

\draw (bd) edge[scheduleplain,out = 170,in = 30] node[nodetime,pos=0.1,label ={above:$U_5\Delta_2$}] {} node[nodeif,pos=0.6,label ={above:$U_2<e(t,s,age)$}]{} node[pos=0.6,rotate=10] {\large{$\int$}} (emi);

\draw (bd) edge[scheduleplain] node[nodetime,pos=0.15,label ={above:$U_6\Delta_2$}] {} node[nodeif,pos=0.48,label ={[align=center]above:$U_3<b(t,age)$\\$\wedge sex=f$}]{} node[pos=0.47,rotate=30] {\large{$\int$}} (birth);

\node[interface,draw=yellow,fill=yellow!10] (ini2) at (5,-8) {Interface $D$\\ to\\simulation\\$(bd',bd',s')$};
\draw (birth) edge[scheduleplain,bend right = 30] node[nodearg,pos=0.5,align=center] {$bd',s'$} (ini2);

\node[eventtype,label={below:$active\leftarrow false$}] (rem) at (11,-7) {Remove};

\draw (death) edge[scheduleplain,bend left = 7] (rem);
\draw (rem) edge[scheduleplain,dashed,line width=0.1pt,bend right = 4] (death);

\draw (emi) edge[scheduleplain,bend left = 5] (rem);
\draw (rem) edge[scheduleplain,dashed,line width=0.1pt,bend right = 0] (emi);

\draw (rem) edge[scheduleplain,dashed,line width=0.1pt,out=30,in=320] (bd);

\draw (rem) edge[scheduleplain,dashed,line width=0.1pt,bend right = 10] (birth);
\end{tikzpicture}
    \caption{\Personagent-layer of the time-update concept of GEPOC ABM using event-graph-like notation. Coloured boxes are the interfaces with the diagrams shown in Figure \ref{fig:update_concept_sim} and \ref{fig:update_concept_gov}. Functions $f_i$ are explained in the text.}
    \label{fig:update_concept_pa}
\end{figure}
\paragraph{\Personagent-layer.}
\important{The dynamics of the \mbox{\personagent} layer, shown in Figure \mbox{\ref{fig:update_concept_pa}}, is defined by the three demographic standard-events birth, emigration and death. These events are decided and scheduled on yearly basis at the \mbox{\pa}s birthday.}
\newpage 
We describe the events in the diagram in temporal order.
\begin{itemize}
\item \textit{Init}:

The \textit{Init} event can be interpreted as the ``Run'' event of the {\pa}'s DE model. Since it is hierarchically inferior to the discrete-time model of the simulation layer, the DE model does only update, if allowed explicitly to do so (see below). 

The \textit{Init} event sets the fixed \textit{birthdate} and \textit{sex} of the agent, which were passed from the simulation layer. Moreover it synchronises the discrete event simulator of the agent with the one from the simulation for the very first time by setting $t=t_{s}$. The agent is furthermore set to \textit{active}, meaning that it will be regarded by the simulation layer's planning and execute events. Moreover, the event computes the agent's age at time $t_s$ via $f_2$. See \ref{sec:parameter_functions} for details on the specification of this function.

Most important feature of the \textit{Init} event is, that it starts the annual birthday cycle. I.e. it calls the \textit{Birthday} event with a time delay of $\Delta_2$. The latter is going to repeat annually until the \pa is removed or the simulation terminates (see below).
The value of $\Delta_2$ is given by function $f_3$ which computes the absolute time durations (in seconds) between the current simulation time $t_s$ and the {\pa}s previous ($\Delta_1$) and next birthday ($\Delta_2$).

Since it is possible, that demographic events may happen before the agent has had it's first birthday-event in the simulation, the \textit{Init} event can already trigger demographic events with a scaled-down likelihood (analogous to the \textit{Birthday} explained below to which we refer for details). The scaling factor $\xi$ is equivalent to the fraction of the life-year remaining until the first \textit{Birthday} event will take place, i.e. $\Delta_2/(\Delta_2+\Delta_1)$.

\item \textit{Birthday}:

\important{As in reality, the model regularly ``celebrates'' birthdays of \mbox{\pa}s. In the model they are used to increment their age and plan/schedule demographic events for the upcoming life-year: Three uniformly distributed random numbers $U_1, U_2,$ and $U_3$ are drawn deciding, whether the agent will die, emigrate or give birth to a new \mbox{\pa} in between the time of the event and the agent's upcoming birthday.}

\important{In this process the random numbers $U_1-U_3$ are compared with the corresponding time, sex and age dependent probabilities $d(t,s,age),e(t,s,age)$ and $b(t,s,age)$.} Note, that these are almost but not fully equivalent to the input parameters $D^p,E^p$ and $B^p$, since a minor correcting transformation is applied before. See Section {\ref{sec:parameter_functions}} for details.

\important{In case any of the events is triggered, a random delay for the event is scheduled. This is done by multiplying the time duration between the current event and the agent's next birthday ($\Delta_2$ as computed by $f_3$) by a uniformly distributed random number, as indicated by $U_4,U_5$ and $U_6$.}

\item \textit{Death}, \textit{Emigration} and \textit{Remove}:

\important{The \textit{Death} and \textit{Emigration} events have the same mechanistic effect on the model, since both cause the agent to leave the scope ($active\leftarrow false$) via the \textit{Remove} event.} Apart from that, the \textit{Remove} event cancels all potentially scheduled events for the future of the \pa from the event queue. In the GEPOC ABM base implementation, this refers to the \textit{Birthday} event, all the potentially scheduled \textit{Birth}, \textit{Emigration} and \textit{Death} event. This is indicated by cancelling edges in the diagram.\footnote{Cancelling edges, by definition, only remove one (the next) potential occurrence of the event from the queue. This is sufficient since no more than one event per type can be scheduled at once.}

\item \textit{Birth}:

\important{The \textit{Birth} event eventually leads to the creation of a new \mbox{\pa} since it triggers the \textit{Add} event in the simulation layer.} The current simulation time is regarded as new birth date, sex is sampled randomly via $f_4$ (see \ref{sec:parameter_functions} for details). In the simulation layer, the \textit{Add} event will be executed as soon as the corresponding DES is updated again. Therefore, agents are only added to the simulation at the discrete time steps of the simulation layer after all individual DESs are finished updating. This, besides other minor problems, prevents unfairness due to the sequence for updating the agents' DESs.

The model is not designed to depict births of twins, triplets, etc.. Hence, the corresponding parameter function must compensate for this.

\item \textit{Time Step Planning}:

While the \textit{Time Step Planning} event is the least interesting event in the GEPOC ABM base version, it is usually one of the most important ones for applications. This event is reserved to schedule \pa events occurring on time-step basis, in particular the ones involving agent-agent interaction. Since these events are scheduled based on the current state of the interacting agents, the modeller must be careful in this process. We highly recommend to schedule agent-agent interaction events with highest priority for time $t_{s1}$. Scheduling them for any later date might cause any of the two agents' states to have changed already, depending on the sequence for updating the agents' DESs.

\item \textit{Time Step Execute}:
 \important{As mentioned earlier, this event solely updates the \pa's DES. That means, that the event queue is processed until the next queued event's scheduling time lies after the passed-on simulation time $t_{s2}$.} Independent of the schedule time of the last processed event, the time variable of the \pa's DES is advanced to $t_{s2}$ afterwards.
\end{itemize}
\newpage
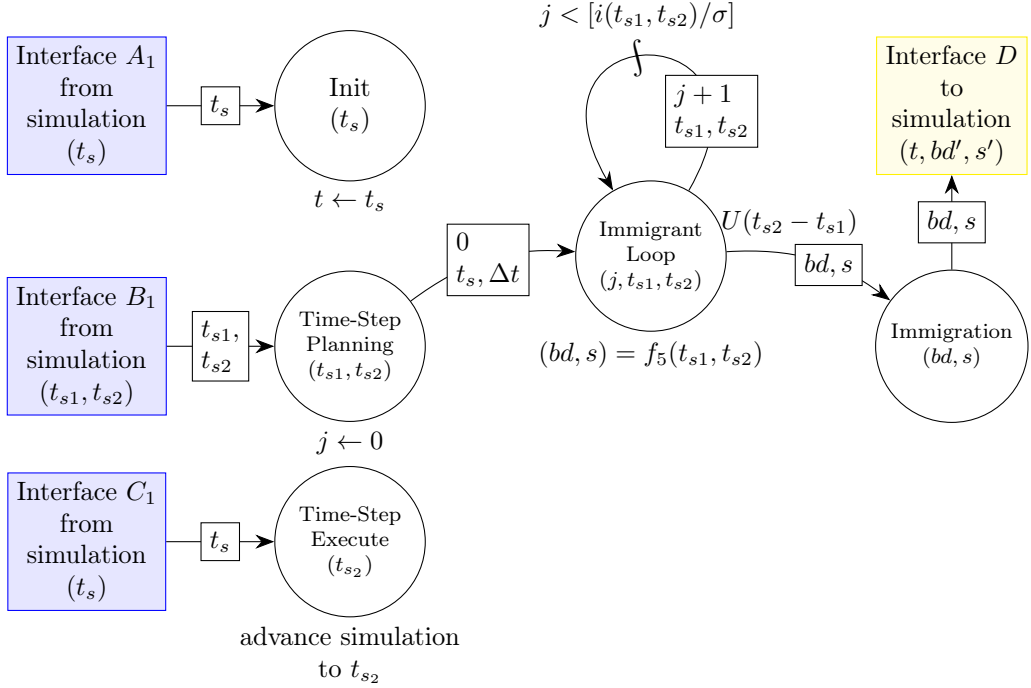
\begin{figure}[H]
    \centering
    \begin{tikzpicture}[
eventtype/.style = {circle, draw, minimum width = 2 cm,text=black,draw=black, align = center},
interface/.style = {rectangle, draw, minimum width = 2 cm,text=black, align = center},
scheduleplain/.style = {-{Stealth[scale=1.5]}},
nodeif/.style = {pos=0.66,anchor=center, circle, minimum width=0.5cm},
nodetime/.style = {pos=0.05,anchor=center},
nodearg/.style = {pos=0.4,anchor=center, rectangle, minimum width=0.5 cm,fill=white,draw},
]

\node[interface,draw=blue,fill=blue!10] (int1) at (-3.5,0) {Interface $A_1$\\ from\\simulation\\$(t_s)$};
\node[interface,draw=blue,fill=blue!10] (int2) at (-3.5,-3.2) {Interface $B_1$\\ from\\simulation\\$(t_{s1},t_{s2})$};
\node[interface,draw=blue,fill=blue!10] (int3) at (-3.5,-5.8) {Interface $C_1$\\ from\\simulation\\$(t_s)$};

\node[eventtype,label={[align=left]below:$t\leftarrow t_{s}$}] (init) at (0,0) {Init\\$(t_s)$};
\node[eventtype,font=\footnotesize,label ={below:$j\leftarrow0$}] (tsp) at (0,-3.2) {Time-Step\\Planning\\$(t_{s1},t_{s2})$};
\node[eventtype,font=\footnotesize,label={[align=center]below:advance simulation\\ to $t_{s_2}$}] (tse) at (0,-5.8) {Time-Step\\Execute\\$(t_{s_2})$};

\draw (int1) edge[scheduleplain] node[nodearg,pos=0.5,align=left]{$t_s$} (init);
\draw (int2) edge[scheduleplain] node[nodearg,pos=0.5,align=left]{$t_{s1},$\\$t_{s2}$} (tsp);
\draw (int3) edge[scheduleplain] node[nodearg,pos=0.5,align=left]{$t_s$} (tse);

\node[eventtype,font=\footnotesize,label ={[align=left]below:$(bd,s)=f_5(t_{s1},t_{s2})$}] (ci) at (4,-2) {Immigrant\\Loop\\$(j,t_{s1},t_{s2})$};
\draw (tsp) edge[scheduleplain,bend left = 20] node[nodearg,pos=0.5,align=left]{$0$\\$t_s,\Delta t$} (ci);

\draw (ci) edge[scheduleplain,out=60,in=120,looseness=7] node[nodearg,pos=0.18,align=left]{$j+1$\\$t_{s1},t_{s2}$} node[nodeif,pos=0.55,label ={above:$j<[i(t_{s1},t_{s2})/\sigma]$}] {} node[nodeif,pos=0.55,rotate=20]{\large{$\int$}} (ci);

\node[eventtype,font=\footnotesize] (immi) at (8,-3.2) {Immigration\\$(bd,s)$};
\draw (ci) edge[scheduleplain,bend left = 20] node[nodetime,pos=0.35,label={above:$U(t_{s2}-t_{s1})$}] {} node[nodearg,pos=0.6,align=center] {$bd,s$} (immi);

\node[interface,draw=yellow,fill=yellow!10] (int4) at (8,0) {Interface $D$\\ to\\simulation\\$(t,bd',s')$};
\draw (immi) edge[scheduleplain] node[nodearg,pos=0.5,align=center] {$bd,s$} (int4);

\end{tikzpicture}
    \caption{\Interfaceagent-layer of the time-update concept of GEPOC ABM using event-graph-like notation. Coloured boxes are the interfaces with the diagrams shown in Figure \ref{fig:update_concept_sim} and \ref{fig:update_concept_pa}. Functions $f_i$ are explained in the text.}
    \label{fig:update_concept_gov}
\end{figure}
\paragraph{\Interfaceagent.} Finally, the processes of the \interfaceagent are described using the diagram displayed in Figure \ref{fig:update_concept_gov}. \important{In the base version of GEPOC ABM, its only purpose is to generate immigrated agents and add them to the simulation.}
\begin{itemize}
\item \textit{Init}

The \textit{Init} event can be regarded as a ``Run'' event for the DE model of the \interfaceagent-layer. Since it is hierarchically inferior to the discrete-time model of the simulation layer, the DE model does only update, if allowed explicitly to do so (see below).  In contrast to the \pa-layer, the \textit{Init} event has no additional purpose for the \interfaceagent.

\item \textit{Time-Step Planning}

The \textit{Time-Step Planning} event is designed to plan all immigration events for the upcoming time-step.

\important{Given a certain time-frame $[t_{s1},t_{s2})$, the \mbox{\interfaceagent} schedules $[i(t_{s1},t_{s2})/\sigma]$ immigration events. Hereby, $i(t_{s1},t_{s2})$ stands for the total number of immigrants to be expected in between $t_{s1}$ and $t_{s2}$, $\sigma$ is the scaling factor of the simulation, and $[\cdot]$ rounds the fraction.} This strategy implicitly assumes that immigrants enter the country uniformly distributed over the course of the year.
In the diagram, this loop is indicated by the \textit{Immigrant Loop} event, which follows the same structure as e.g. the \textit{Init Loop} event in the simulation layer. For details regarding functions $i$ and $f_5$, we refer to Section \ref{sec:parameter_functions}.

\important{Every iteration, the event samples a random birth-date and sex for the immigrated agent.} This is indicated by $f_5$ and follows the same principle as within the initialisation of the initial agent population, denoted as $f_1$ in the simulation layer. \important{In addition, the event schedules a random immigration date.} This is done by sampling a uniformly distributed date between $t_{s1}$ and $t_{s2}$ and reflects the assumption stated earlier\footnote{Depicting processes like immigration in DE models would typically be done via inter-arrival rates. Hereby an immigration event would schedule itself with exponentially distributed delay. We chose to use the described different strategy since it (a) reduces the variance of the total number of immigrants (it is zero essentially) and (b) is easier to integrate in the overall time-tick update scheme of the simulation layer.}.

\item \textit{Immigration}

The \textit{Immigration} event triggers the \textit{Add} event in the simulation layer via the corresponding interface. As a result, immigrants enter the model at the same time as births and are first regarded in the next \textit{Time-Step Planning} event of the simulation layer.
\end{itemize}

\subsection{Design Concepts}
\subsubsection{Basic Principles}
There are a couple of key principles, which were motivation for the rather unusual structure of the model.

\paragraph{Individual DESs.} First of all, the use of individual DESs for simulation and agents seems odd at the first glance. Why not put every event into a global queue? Surprisingly, this strategy comes with many benefits, in particular thinking about the implementation. 
\begin{itemize}
    \item The update of the {\pa}s' and the {\interfaceagent}'s DES can be made in parallel, since they only interact with each other via the simulation layer and at the predefined time-ticks. This strategy can be interpreted as co-simulation (e.g. \cite{hafner2017terminology}).
    \item Long event queues are the crucial factors for performance problems of DES. In the presented model, neither of the \pa or simulation layer DE queues ever exceeds the length of four. This is a huge benefit compared to a model with one huge event queue, which entries would scale with the number of agents.
\end{itemize}

\paragraph{Planning and execute.} Time update is ultimately driven by the iterative use of \textit{Time-Step Planning} and \textit{Time-Step Execute}. This concept is highly beneficial to maintain a logically correct order of processes and to avoid any bias problems occurring due to the looping order of the agents. Although this strategy does not show its full potential in the base version of GEPOC ABM as a result of missing interactions, it proves to be very valuable in applications extending the model. 

\paragraph{Birthday events.} The use of \textit{Birthday} events to plan the future life-year of the \pa might seem unusual, but is motivated by the internationally used standard-definition of ``death probability'' as used by various national census institutes. E.g. 
\begin{center}
 ``The probability of death at some age x refers to the probability of a person living until the age of x to die during that year of age.'' (Statistics Finland, \href{https://www.stat.fi/meta/kas/kuolemanvaara_en.html}{https://www.stat.fi/meta/kas/kuolemanvaara\_en.html})
\end{center}
Therefore, death probabilities from corresponding institutions can directly be used for parametrisation. Emigration and birth probabilities can be estimated from yearly births and emigrants using the same methods.

As a small but fine added value on top: the \textit{Birthday} event is also used to increment age. This way, age must not be calculated from the agent's birth-date every time it is needed, but can be taken directly.

\subsubsection{Interaction}
The model does support interaction between model {\pa}s via corresponding interfaces at discrete points in time. Yet, in the base implementation, this feature is not used.

\subsubsection{Stochasticity}
The model uses stochasticity at various points. In the newest version, the age, sex and regional distribution of the initial agent population is not subject to randomness anymore (compare with \cite{bicher_gepoc_2018,bicher_definition_2015}) which is helpful to reduce variance. Yet basically all planned \pa-events involve stochastic decisions and stochastic scheduling dates.

As the ABM involves (not-even) mean-field interaction between the agents, the relative standard deviation to decrease by $\frac{1}{\sqrt{N}}$(compare \cite{bicher_classification_2017}). Since the model is usually run with several million {\pa}s, between five to ten Monte-Carlo iterations (depending on the volatility of the variable of interest) are typically sufficient to well estimate the mean value of the aggregated numbers, i.e. number of agents with certain age and sex (compare \cite{bicher_review_2019}). 

\subsubsection{Observation}
\label{sec:basemodel_observation}
\important{As an additional plus, the time-step based update of the overall model helps with model observation of the model states. The model distinguishes between two types of output variables: snapshot and differential.} The former tracks the current state of the simulation in between any of the simulation's time-steps (done in the context of the \textit{Observe} event, see Figure \ref{fig:update_concept_sim}), the latter tracks the change of the model by counting events executed within a model time-step (i.e. the events triggered and executed in the context of the \text{Time-Step Execute} events, see Figure \ref{fig:update_concept_sim}). This distinction is not only helpful for output analysis, but also helps with verification and validation of the model (i.e. check if snapshot $+$ differential $\overset{!}{=}$ new snapshot).

\important{To specify which output variables are possible, we first introduce the logic of \textit{agent-characteristics}.} 
\definition{Agent-Characteristic}{def:characteristic}{
A mapping $\Lambda^k$ is said to \textbf{characterise} the agent population with respect to \textbf{characteristic} $k$, if it operates on the joint state space $S$ of the agent of the specified type (usually the \pa type), and maps the agent's temporal state  $a(t)$, onto zero or one:
\begin{equation}\label{eq:characteristics}
    \Lambda^k: S\rightarrow \{0,1\}: a,t\mapsto \Lambda^k(a(t)).
\end{equation}
}
The most obvious choices for such functions would be age- or sex-assessments, i.e. the function returns one, if and only if the current age or sex of the \pa is equal to the specified value or lies within a specified age interval. 

Clearly this concept is helpful for specifying the model output. With $K$ different characteristics $\{1,\dots,K\}$, the \textit{snapshot-output} of the model is given by
\[O_k(t_i):=\sum_{j=1}^{N}\Lambda^k(\pa_j(t_i)).\]
The differential output can also make use of this concept by counting only events of agents which fulfil a certain characteristic.

Here we state some possible outcome variables using the defined concept:
\begin{itemize}
    \item Total number of {\pa}s (with age $a$, sex $s$, birth-date $bd$) at time $t_i$ (\textit{snapshot-output}).
    \item Total number of died, emigrated, immigrated {\pa}s (with age $a$, sex $s$, birth-date $bd$) in between $t_i$ and $t_{i+1}$ (\textit{differential-output}).
    \item Total number of newborn {\pa}s or ``mother'' {\pa}s (with age $a$, sex $s$, birth-date $bd$) in between $t_i$ and $t_{i+1}$ (\textit{differential-output}).
\end{itemize}
One must be aware, that increasing the number of tracked outcomes negatively influences the computation time. So we clearly recommend choosing the characteristics in a minimalist fashion.

\subsection{Details}
By using the event graph standard and the added explanation, the model definition is not yet fully reproducible. Some technical aspects of the parameter functions remain to be discussed, in particular with respect to model parametrisation.

\subsubsection{Initialisation}
\label{sec:initialisation}
Since we used the Event Graph notation for describing the model, it actually lacks a separate initialisation part. All processes which would usually be referred to in the initialisation are already described in the process overview in Section \ref{sec:basemodel_process}. The presented model definition does not distinguish if an agent is created in the course of the \textit{Init Loop} at $t_0$ or at any later point in the course of the simulation -- both are considered to be a part of the model dynamics. 

\subsubsection{Submodels}
\label{sec:parameter_functions}
Formally correct treatment of ``time'' in GEPOC ABM is not simple, since both model parametrisation and update require a date-time representation of it. The model internally uses SI unit seconds and any time-value within the model can be regarded as total number of seconds elapsed since 1970-01-01 (UNIX time). For date-time representation we use the ISO time tuple $t_0\cong y_0\text{-}m_0\text{-}d_0TH_0\text{:}M_0\text{:}s_0$ using the rules of UTC time. This way we establish an isomorphism between model-time and date-time. We conventionally write points in time with $t$ using a suitable index identifier $i$ and write
\[t_i\cong (y_i,m_i,d_i).\]
Note that we drop higher resolved components of the tuple to keep the documentation readable. Anyway, this concept allows us to define subtraction and addition between time-tuples:
\[(y_i,m_i,d_i)\pm(y_j,m_j,d_j) = (y_i,m_i,d_i)\pm t_j = t_i \pm(y_j,m_j,d_j) = t_i \pm t_j.\]

Using this notation, we are finally capable of adding details to the parameter functions introduced earlier. Hereby we refer to the five helper functions $f_1$ to $f_5$ and the demographic functions $p,i,d,e$ and $b$. We will conventionally use variable name $t\cong (y,m,d)$ for time, $bd\cong(y_{bd},m_{bd},d_{bd})$ for birth-date, $a$ for age, and $s$ for sex.

\noindent\paragraph{Functions $p$ and $f_1$.} First of all, $p(t)$ stands for the total population at time $t$. Considering that most census data is given on yearly base, i.e. $P(y)$ stands for the population of the region at $(y,1,1)$, we define this function as
\begin{equation} p:t\mapsto p(t)=\left\lbrace\begin{array}{l}P(y),\text{   if   } (t-(y,1,1))<((y+1,1,1)-t), \\
P(y+1),\quad \text{otherwise.}
\end{array}\right.\end{equation}
Note, that the switch case takes the population from the closest ``first-of-first'' of a year.

Moreover, sampling of birth-date and sex is a two step process: $f_1(t)=f_{11}(t)\circ f_{12}$. First, a random age $a$ and sex $s$ are drawn from a joint age-sex distribution for the given time:
\begin{equation}f_{11}:t\mapsto f_{11}(t)=(a,s)=(X,Y) \text{   with   }Pr(X=a,Y=s|t)=\left\lbrace\begin{array}{l}\frac{P(y,s,a)}{P(y)},\text{   if   } (t-(y,1,1))<((y+1,1,1)-t), \\ \frac{P(y+1,s,a)}{P(y+1)},\quad \text{otherwise.}\end{array}\right.\end{equation}
In a second step, a random birth-date $bd$ is sampled under the assumption, that births are distributed uniformly within the course of the year:
\begin{equation}f_{12}:(a,s)\mapsto f_{12}(a,s)=(bd,s)=\left((y-a,1,1)+U\cdot ((y-a+1,1,1)-(y-a,1,1)),s\right).\end{equation}
In the newest version of GEPOC ABM, $p$ and $f_{11}$ are no longer separate processes. Instead of creating a random sex and age for all $[p(t)/\sigma]$ agents, we would instead create
\begin{equation}[P(t,s,a)/\sigma]_{s},\text{    with    }P(t,s,a)=\left\lbrace\begin{array}{l}P(y,s,a),\text{   if   } (t-(y,1,1))<((y+1,1,1)-t), \\ P(y+1,s,a),\quad \text{otherwise,}\end{array}\right.\end{equation}
agents with sex $s$ and age $a$ for all age and sex combinations supported by the parametrisation. This not only avoids (most) fluctuations for the population distribution at starts, it also is computationally less expensive. 

Anyway, since numbers may become small here, we use a specific stochastic rounding method $[\cdot]_{s}$ which is explained below.
\noindent\paragraph{Functions $i$ and $f_5$.}
The concept for immigration is very similar to the one for creating the initial population, yet we have to care for time-differences instead to absolute points-in-time here. We split the interval $(t_{s1},t_{s2})$ into the smallest number $n$ of disjoint sub-intervals, so that the start and endpoint of the interval have equal year:
\begin{equation}
    [t_{s1},t_{s2})=:\dot{\bigcup}_{i=1}^n[t_{i}^s,t^e_{i})=:\dot{\bigcup}_{i=1}^n[(y_i,m_i^s,d_i^s),(y_i,m_i^e,d_i^e)).
\end{equation}
If $t_{s1}$ and $t_{s2}$ lie within the same year, clearly $n=1$, $t_1^s=t_{s1}$, $t_1^e=t_{s2}$. Otherwise, $t_1^s=t_{s1}$, $\forall i>1: t_i^s=(y+i-1,1,1)$, $\forall i<n: t_i^e=(y+i,1,1)$, $t_n^e=t_{s2}$ is the minimalist solution. In any case, $y_i=y+i-1$. Furthermore define
\begin{equation}
    \delta_i=\frac{t^e_{i}-t_i^s}{(y_i+1,1,1)-(y_i,1,1)}.
\end{equation}
With this notation, we finally compute the parameter function. Let $I(y_i)$ stand for the total number of immigrants between $(y_i,1,1)$ and $(y_i+1,1,1)$, then
\begin{equation}i:(t,\Delta t)\mapsto i(t,\Delta t)=\sum_{i=1}^{n}\delta_iI(y_i).\end{equation}

Similar to the initialisation of the start population, sampling of birth-date and sex is a two step process: $f_5(t)=f_{51}(t)\circ f_{52}$. Let $I(y_i,s,a)$ stand for the total number of immigrants with sex $s$ and age $a$ within year $y_i$, then
\begin{equation}f_{51}:(t,\Delta t)\mapsto (s,a)=(X,Y)\text{    with    }Pr(X=s,Y=a|t)=\frac{\sum_{i=1}^{n}{\delta_i\frac{I(y_i,s,a)}{I(y_i)}}}{\sum_{i=1}^{n}\delta_i}.\end{equation}
Furthermore $f_{52}$ is equivalent with $f_{12}$.

Analogous to the start population, we couple $i$ and $f_{51}$ in the newest version of GEPOC ABM. Hereby the mentioned stochastic rounding becomes even more valuable since numbers for very old and very young immigrants can become quite small.

\noindent\paragraph{Function $f_4$.} This function is used to sample the biological sex of a newborn \pa. It is modeled as a random boolean decision:

\begin{equation}
f_4:t\mapsto f_4(t)=s=\left\lbrace\begin{array}{l}\text{male},\text{  if } U<\frac{B(y,\text{male})}{B(y)}\\ 
\text{female}, \text{  else.}\end{array}\right.
\end{equation}
Hereby, $B(y,m)$ corresponds to the number of newborn males in the course of year $y$ (See Definition \ref{def:sex} for interpretation of the \textit{sex} variable). Typically, this fraction is rather country specific and is very stable with time. Therefore, in the current model version, it is parametrise it by one constant value $\alpha_m$:
\[\forall y: \frac{B(y,\text{male})}{B(y)}\approx \alpha_m.\]

\noindent\paragraph{Functions $b,d,e$.}
To parametrise the parameter functions $b,d$ and $e$, we have to deal with time-intervals again, yet in contrast to function $i$ we do not have to deal with potential problems caused by the outermost step-size $\Delta i$ due to their definition as 
\begin{center}
    \textit{probability, that the corresponding event occurs to someone with sex $s$ who aged $a$ in year $y$ before the persons $a+1$-st birthday.}
\end{center}
Since such numbers are often given directly by census institutions and are provided on yearly basis, we define
\begin{equation}
    \left(\begin{array}{c}d(t,s,a)\\
    e(t,s,a)\\b(t,a)\end{array}\right)=\Psi\left(\begin{array}{c}D^p(y,s,\min(a,a_{max}))\\E^p(y,s,\min(a,a_{max}))\\
    B^p(y,\min(a,a_{max})\end{array}\right)
\end{equation}
whereas $B^p,D^p,E^p$ stand for the corresponding parameter value valid between $(y,1,1)$ and $(y+1,1,1)$, and $a_{max}$ is the highest single-age class regarded by the model parameters. Function $\Psi$ represents the correction transformation from Theorem \ref{thm:aposterior_aprior} (we refer to Section \ref{sec:aposterior_aprior} for details) and removes bias due to simultaneous events. Note that sex is no argument in $B(\cdot,\cdot)$, since this process only targets female {\pa}s (See Definition \ref{def:sex} for interpretation of the \textit{sex} variable).

\noindent\paragraph{Functions $f_2$ and $f_3$.}
These two functions are simple but not trivial helper routines to compute agent-specific variables related to the agent's birth-date $bd$, birth-day, and the current time $t$.

First $f_2$ computes the agent's current age $a$ in years. We get
\begin{equation}
f_2:(bd,t)\mapsto f_2(bd,t)=a = f_2(bd,t_{s})=\left\lbrace\begin{array}{l}y-y_{bd},\quad m_{bd}<m\\ 
y-y_{bd},\quad m_{bd}=m\wedge d_{bd}\leq d\\
y-y_{bd}-1,\quad m_{bd}=m\wedge d_{bd}>d\\
y-y_{bd}-1,\quad m_{bd}>m\end{array}\right.
\end{equation}
This cumbersome computation is due to the fact that ``year'' is not a proper time unit (leap-days/seconds) and it becomes even more cumbersome, if hours, minutes, \dots are regarded as well.

The second routine $f_3$ computes the position of the current time within the birthday cycle of the agent. The output is a two element vector, whereas the first entry gives the time difference between the current date and the {\pa}s last birthday and the second entry gives the time difference between the agents next birthday and now.  Given the current age $a$ of the agent at time $t$, we get $bd_{age+1}:=(y_{bd}+age+1,m_{bd},d_{bd})$ as the next and $bd_{age}:=(y_{bd}+age,m_{bd},d_{bd})$ as the prior birth-day with $bd_{age}\leq t<bd_{age+1}$ .
Consequently we get
\[f_3:(bd,t)\mapsto f_3(bd,t)=(t-bd_{age},bd_{age+1}-t).\]

\paragraph{Stochastic Rounding $[\cdot]_s$.} In particular, when small numbers are scaled down and rounded, we issue numerical problems since we round to $0$ disproportionately often.\footnote{We give an example for this problem considering a total population $P(t)=1000$ and using a scaling factor $\sigma=100$. Clearly we expect that $10$ agents in total are generated by the model. Furthermore, the population is split into $100$ age-groups, containing $10$ persons each. Down-scaling $10$ by $\sigma$ would result in $0.1$ agents per age-group. Using classic rounding, we would initialise $0$ agents for every age-group, which results in a (wrong) total agent population of zero.}

To solve this problem, we introduce the following stochastic rounding strategy $[\cdot]_s$:
\begin{equation}[x]_s:=X,\text{ with } Pr(X=\lfloor x\rfloor+1)=x-\lfloor x\rfloor,\quad Pr(X=\lfloor x\rfloor)=1-(x-\lfloor x\rfloor)\end{equation}
That means, the probability that a number is rounded up is its decimals.

This way, sums of rounded summands is, in expectation, equivalent with the rounded sum. Therefore, this strategy helps preventing de-aggregation problems, as the one introduced earlier.

\subsubsection{Summary: Model Parameters}
\label{sec:parameters}
In this section we summarise the parameters used in the model and hereby display the demand for parametrisation. Note that we do not specify how the corresponding parameter values can be found.\footnote{Therefore this section is not called ``Input Data'' as recommended by the ODD protocol} General parameters are found in Table \ref{tbl:params1}, demographic parameters in Table \ref{tbl:params2}.
\begin{table}[H]
    \caption{General parameters / model input of GEPOC ABM.}
    \label{tbl:params1}
    \begin{center}
\begin{tabular}{p{5cm}|c|c|c|p{3cm}}
\hline
Parameter & Dimensions & Unit & Parameter Space   & Interpretation\\
\hline
$t_0=y_0\text{-}m_0\text{-}d_0TH_0\text{:}M_0\text{:}s_0$ & - & date-time & date-time-space &start date-time of the simulation\\
\hline
$\Delta t_i$ & $i=1,\dots,m$ & seconds & $\mathbb{R}^+/\{0\}$ &  time-tick lengths\\
\hline
$t_e=y_e\text{-}m_e\text{-}d_eTH_e\text{:}M_e\text{:}s_e$ $:=y_0\text{-}m_0\text{-}d_0TH_0\text{:}M_0\text{:}s_0+ \sum_{i=1}^m \Delta t_i$& - & date-time & date-time-space & end date-time of the simulation\\
\hline
$\sigma$ & - & - & $\mathbb{R}^+/\{0\}$ &  scaling factor of the simulation\\
\hline
\end{tabular}
\end{center}
\end{table}

\begin{table}[H]
    \caption{Demographic parameters of GEPOC ABM. See Definition \ref{def:sex} for interpretation of the \textit{sex} variable.}
    \label{tbl:params2}
    \begin{center}
\begin{tabular}{p{1.5cm}|p{5.5cm}|c|c|p{4cm}}
\hline
Parameter & Dimensions & Unit & P. Space   & Interpretation\\
\hline
$\alpha_m$ & - &  probability & $[0,1]$  & probability for male \pa at birth\\
\hline
$a_{max}$ & - & years & $\mathbb{N}/\{0\}$ & maximum age regarded in the parameters\\
\hline
$P(y,s,a)$ & $y\in \{y_0,\dots,y_e\}$, $a\in \{0,\dots, a_{max}\}$, $s\in \{\text{male},\text{female}\}$ & persons & $\mathbb{N}\cup\{0\}$ &  total population per age $a$, sex $s$ at the start of year $y$.\\
\hline
$I(y,s,a)$ & $y\in \{y_0,\dots,y_e\}$, $a\in \{0,\dots, a_{max}\}$, $s\in \{\text{male},\text{female}\}$ & persons & $\mathbb{N}\cup\{0\}$ &  total immigrants with age $a$ (at time of immigration), sex $s$ in the course of year $y$.\\
\hline
$D^p(y,s,a)$ & $y\in \{y_0,\dots,y_e\}$, $a\in \{0,\dots, a_{max}\}$, $s\in \{\text{male},\text{female}\}$ & probability & $[0,1]$ &  Probability of a person with sex $s$, who has had its $a$-th birthday in year $y$, to die before its $a+1$-st birthday.\\
\hline
$E^p(y,s,a)$ & $y\in \{y_0,\dots,y_e\}$, $a\in \{0,\dots, a_{max}\}$, $s\in \{\text{male},\text{female}\}$ & probability & $[0,1]$ &  Probability of a person with sex $s$, who has had its $a$-th birthday in year $y$, to emigrate before its $a+1$-st birthday.\\
\hline
$B^p(y,s,a)$ & $y\in \{y_0,\dots,y_e\}$, $a\in \{0,\dots, a_{max}\}$, $s\in \{\text{male},\text{female}\}$ & probability & $[0,1]$ &  Probability of a female person, who has had her $a$-th birthday in year $y$, to give birth to a child before her $a+1$-st birthday. This probability must compensate for multiple-births which are not depicted in the model.\\
\hline
\end{tabular}
\end{center}
\end{table}

\newpage
\section{GEPOC ABM Geography - Model Definition}
\label{sec:geographymodel}
\important{GEPOC ABM Geography is a direct extension of Extending GEPOC ABM, as defined in Section \mbox{\ref{sec:basemodel}}, by regional features}. This extension comes with various challenges regarding parametrisation.
We will build on the existing blocks of the ODD protocol from Section \ref{sec:basemodel} and extend and/or modify accordingly.
\subsection{Overview}

\subsubsection{Purpose and Patterns}
We may use this model extension for any kind of research question related to regional distribution and regional change of the population. GEPOC ABM Geography does not depict internal migration processes which poses a clear limitation for its applicability for dynamic research problems.

\subsubsection{Entities, State Variables and Scales}
\label{sec:geography_entities}
In addition to date of birth and biological sex, {\pa}s are given a static
\begin{itemize}
    \item \important{geographical coordinate, in form of longitude and latitude,}
\end{itemize}
which models the {\pa}'s point of residence. We henceforth use variables $(long, lat)$ to describe it.

\important{Given a certain regional-level we can match this point uniquely to a regional identifier. That means we can match the {\pa} uniquely to a certain city, municipality, district, \dots.} To formalise this principle, we introduce the following two definitions.

\definition{Regional-Level, Region-Family and Region-Mapping}{def:regional}{
A family of sets $(A^r_i)_{i=1}^{q}$, which
\begin{itemize}
    \item are $\subset \mathbb{R}^2$,
    \item have finite area,
    \item are pairwise disjoint, and
    \item cover, in total, the full area of interest,
\end{itemize}
is furthermore called \textbf{\regionfamily} and is identified by its joint \textbf{\regionallevel} $r$.
Due to the properties of the family we can define the unique \textbf{\regionmapping}
\begin{equation}\phi:(long, lat,r)\mapsto \phi(long, lat,r) := [i \Leftrightarrow (long, lat)\in A_i^r].\end{equation}
which maps a geo-coordinate to the index, furthermore called \textbf{\regionid}, of the region of the family in which it lies in.
}
\newpage 
\definition{Fineness}{def:fineness}{A \regionallevel $r$ is said to be \textbf{finer} than a \regionallevel $r'$ if $\forall i\in\{1,\dots,q_{r'}\}$ there exists a subset $J_i\subseteq \{1,\dots,q_{r}\}$ so that
\begin{equation}
    A^{r'}_i=\dot{\bigcup_{j\in J}}A^r_j,
\end{equation}
and at least one of the $J_i$ has more than one element.
}
Note that this definition of fineness generates a half order, but not a full order, on the set of all possible {\regionfamilies}. For example, while municipalities are strictly finer than political districts in Austria, they cannot be compared to ZIP codes.

\subsubsection{Process Overview and Scheduling}
Since the model does not add any new processes, the general model update strategy is completely unchanged. Yet, certain parameter functions use additional input variables and generate additional outcomes.

\paragraph{Simulation layer.}

On simulation layer, most importantly, function $f_1^{new}$ with
\begin{equation}
    f^{new}_{1}(t):=g_1(f_{1}(t),t) = (age,s,long, lat)
\end{equation}
replaces $f_1$ and now returns a third and fourth output: longitude $long$ and latitude $lat$ of the \pa residence. It hereby uses \regionallevel
\begin{equation}
    r_{0},\text{   with  \regionfamily } (A_i^{r_0})_{j=0}^{q_i},
\end{equation}
given by the parametrisation of the model.
\important{The corresponding sampling algorithm is the heart of the geography extension and is, in detail, explained in Section {\ref{sec:parameter_functions_geo}}. It is, in general, split into two steps: First a specific region (\mbox{\regionid}) from the \mbox{\regionfamily} is drawn. In a second step, a coordinate within the region is sampled.}

Anyway, the sampled coordinate is, along with birth-date and sex passed on as third and fourth parameter to \textit{Interface $A_2$} and, consequently, the \textit{Init} event of the new created \pa. Analogously, also \textit{Inferface $D$} and the \textit{Add} event are extended to four parameters.

\paragraph{\Personagent layer.} 

\important{The {\pa} is initialised with two additional arguments \textit{long} and \textit{lat}. While these two arguments are not directly influential for the dynamics, they yet imply \mbox{\regionids} for specific \mbox{\regionallevels} which are used to compute the event-probabilities. We define}
\begin{equation}
    r_{d},r_{e},r_b,\text{   and   } (A_i^{r_d})_{i=1}^{q_d},(A_i^{r_e})_{i=1}^{q_e},(A_i^{r_b})_{i=1}^{q_b}
\end{equation}
\important{as \mbox{\regionallevels} and corresponding \mbox{\regionfamilies} used as spatial resolution to compute death, emigration and birth probabilities. That means, the regional affiliation of an agent w.r. to these regions is relevant for computing the corresponding event probability for death, emigration or birth.} They can, but do not need to differ and should be chosen suitable for the quality and resolution of parametrisation data.

With these identifiers, we replace the death, emigration and birth probability parameter function $d,e,b$ as follows:
\begin{align}
  d(t,s,age)&\rightarrow d_{new}(t,\phi(long,lat,r_d),s,age),\\
  e(t,s,age)&\rightarrow e_{new}(t,\phi(long,lat,r_e),s,age),\\
    b(t,age)&\rightarrow b_{new}(t,\phi(long,lat,r_b),age).
\end{align}

Finally, the \textit{Birth} event automatically generates the corresponding newborn agent at the same location as the mother \pa. That means, $long$ and $lat$ are inherited to the {\pa}'s offspring.

\paragraph{\Interfaceagent layer.} 

Analogous to the simulation layer, also the \interfaceagent layer uses a changed parameter function
\begin{equation}
    f^{new}_{5}(t,\Delta t):=g_5(f_{5}(t,\Delta t),t,\Delta t) = (age,s,long,lat)
\end{equation}
and passes four instead of two parameters to the corresponding \pa interface. It hereby uses the \regionallevel
\begin{equation}
    r_{i},\text{   with  \regionfamily   } (A_j^{r_i})_{j=1}^{q_i}.
\end{equation}

\subsection{Design Concepts}
\subsubsection{Basic Principles}
The key principle for the geographical extension of GEPOC ABM is to extend the state of every \pa by the two additional variables latitude and longitude, and to use these two variables to compute various regional identifiers. This strategy was chosen in favour adding regional identifiers directly as an attribute, since it is more robust w.r. to extensions and to temporally changing regional structures.

\subsection{Details}
\subsubsection{Submodels}
\label{sec:parameter_functions_geo}
\paragraph{Functions $g_1$ and $g_5$.}
Given the output of $f_1$, i.e. a sampled age and sex, function $g_1$ uses a two step strategy to sample a statistically representative location. i.e. $g_1(t,s,a)=g_{11}\circ g_{12}$.

First of all, $g_{11}$ uses the \regionallevel $r_0$ specified for initialisation to sample a statistically representative region from $(A^{r_0}_i)_{i=1}^{q_0}$. Let $t\cong(y,m,d)$ and $P(y,i,s,a)$ stand for the total population of region $A^{r_0}_i$ with age $a$ and sex $s$ at the beginning of year $y$, then
\begin{equation}
    g_{11}(t,s,a) = (t,X,s,a),\text{   with   } Pr(X=i)=\left\lbrace\begin{array}{l}\frac{P(y,i,s,a)}{P(y,s,a)},\text{   if   } (t-(y,1,1))<((y+1,1,1)-t), \\ \frac{P(y+1,i,s,a)}{P(y+1,s,a)},\quad \text{otherwise.}\end{array}\right.
\end{equation}
In the newest versions of GEPOC ABM Geography, we use census data directly to create statistically representative agents. I.e. for every age $a\in \{0,\dots,a_{max}\}$, sex $s\in\{\text{male},\text{female}\}$ and region $i\in\{0,\dots,q_0\}$ we create
\begin{equation}
    [P(y,i,s,a)/\sigma]_{s}
\end{equation}
agents with the corresponding features (See Definition \ref{def:sex} for interpretation of the \textit{sex} variable). 

It remains to sample a random birth-date (via $f_{12}$, see Section \ref{sec:parameter_functions}) and a coordinate via $g_{12}$.

In the second step, a coordinate within region $A^{r_0}_i$ is drawn. This process founds on the one presented in \cite{bicher_gepoc_2018} and in \cite{gallagher2018spew} (Section 3.3.1), and is extended by a rejection-method using a much finer set-family $(A^{r_{min}}_i)_{i=1}^{q_{min}}$ with $q_{min}\gg q_0$, and a labelling function
    \begin{equation}
        \psi: \{0,\dots,q_{min}\}\rightarrow \{\text{true},\text{false}\}: j\mapsto \psi(j).
    \end{equation}
which labels, whether the corresponding fine-resolved region is inhabited. Typical candidates for the fine regional resolution are raster maps, which are labelled for being inhabited via satellite images and machine-learning.

Furthermore, the algorithm for computing $g_{12}$ is described in two steps:
\begin{enumerate}
\item Draw a uniformly distributed point $(long,lat)$ within the region $A_i^{r_0}$, which was chosen to become the residence region for the agent. For this we may exemplary use the algorithm presented in \cite{bicher_gepoc_2018} based on triangulation. 
\item Furthermore calculate $\phi(long,lat,r_{min}))$ to find, in which of the regions from $(A^{r_{min}}_i)_{i=1}^{q_{min}}$ the point lies in. In case the result of
\[\psi(\phi(long,lat,r_{min})))\]
is true, and the sampled point lies in an inhabited region, the algorithm terminates and $(long,lat)$ is returned. Otherwise, repeat with step 1.
\end{enumerate}
This algorithm extends the one presented in \cite{bicher_gepoc_2018} by the rejection strategy in step 2 and drastically improves the quality of the result. In Figure \ref{fig:points_new}, left, we see 1 Million agents sampled based on municipality data in Austria with the old algorithm from \cite{bicher_gepoc_2018}. Using the Global-Human-Settlement raster layer\cite{ghs_hp} for the rejection strategy presented here, we receive the right part of Figure \ref{fig:points_new}. Comparisons with e.g. satellite images of Austria at night reveal, that this image much more properly represents the topological structure of Austria, in particular with respect to the Alps in the west. A further refinement of the coordinate sampling $g_{12}$ is currently in progress. Here, OpenStreetMap building data will be used to obtain a set of building coordinates for each region, together with an approximate probability distribution for the likelihood of a person from that region living in a given building. The coordinates of each region will then be sampled from this distribution.

Finally, the strategy is analogously extended to compute $g_{5}$, which is the function used to sample residence places for immigrants.

\begin{figure}[H]
\begin{center}
\includegraphics[width=0.4\linewidth]{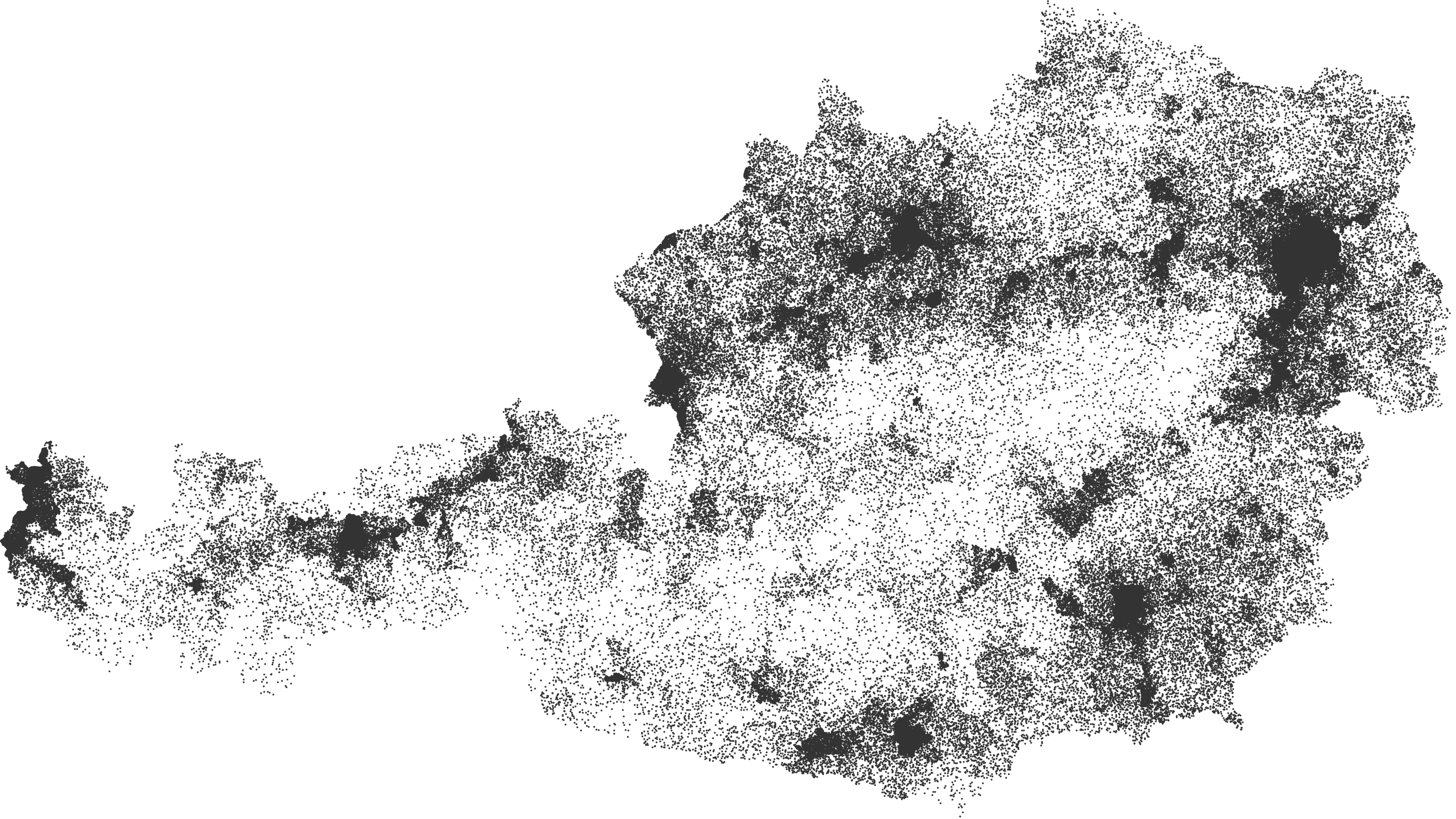}
\includegraphics[width=0.4\linewidth]{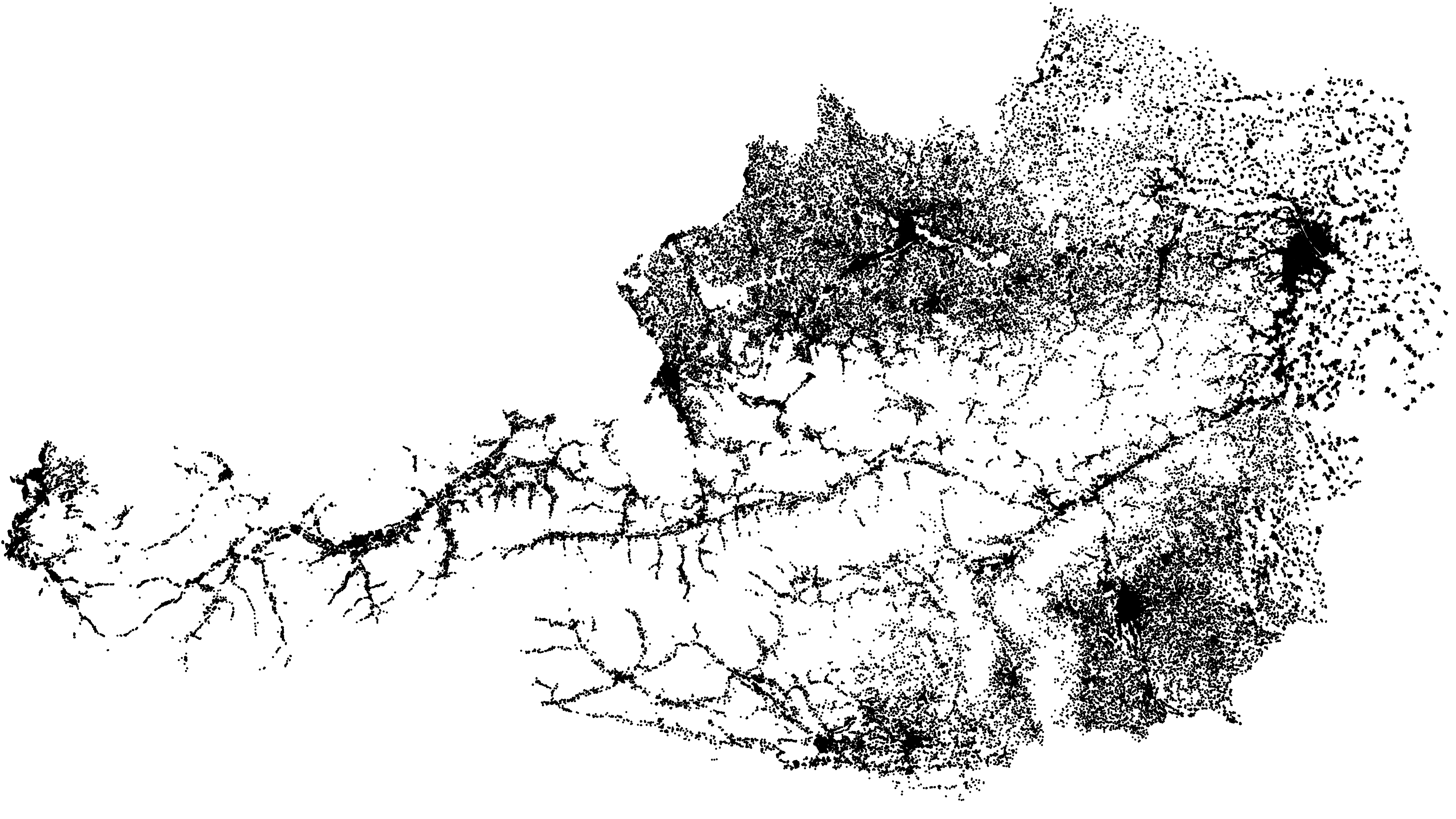}
\caption{Left, sampled residences for $1M$ agents according to the distribution for municipalities in Austria (as $r_{0}$) without rejection sampling, i.e. using only step 1 in the presented algorithm, right, with rejection-sampling method with the Global Human Settlement layer\cite{ghs_hp} as $r_{min}$, i.e. iterating steps 1 and 2 as specified in the presented algorithm. Inhabited and uninhabited regions are much accurately displayed.}
\label{fig:points_new}
\end{center}
\end{figure}

\paragraph{Functions $d_{new},e_{new}$ and $b_{new}$.}
The parameter functions for computing the probabilities are extended accordingly by an additional spatial parameter. With $t=(y,m,s)$,
\begin{equation}
    \left(\begin{array}{c}d_{new}(t,i,s,a)\\
    e_{new}(t,i,s,a)\\b_{new}(t,i,a)\end{array}\right)=\Psi\left(\begin{array}{c}D^p(y,i,s,\min(a,a_{max}))\\E^p(y,i,s,\min(a,a_{max}))\\
    B^p(y,i,\min(a,a_{max})\end{array}\right)
\end{equation}
whereas $D^p,E^p,B^p$ stand for the corresponding parameter value for year $y$, \regionid $i$, sex $s$ and age $a$, and function $\Psi$ corrects the bias due to simultaneous events (see Theorem \ref{thm:aposterior_aprior}). Note that different \regionallevels could be used for parametrisation, the parameter values must be brought to a common finest level before applying $\Psi$ though. 

\subsubsection{Area-Status}
Looking at the continuous change of the political landscape it is worth mentioning that regional set-families might only be valid for a certain time-frame. E.g. in 2015 various districts and municipalities in Austria were fused due to administrative reasons.


Since dynamically adapting to different regional structures would be too difficult w.r. to implementation and parametrisation (e.g. parameter tables are no longer ``rectangular''), we use a static concept: We fix one so called \textit{area-status} (\textit{Gebietstand}, in German) for the simulation, meaning that the spatial reference of the simulation is always given by this sole regional structure -- input and output. Since regional structures are usually updated on yearly basis, we typically identify the status with the year for which it is valid. Problem for this strategy is, that all parameters and consequently all parametrisation data must be given in reference to this status -- independent of the time component of the parameter. For example, in a simulation between 2010 and 2040 with area-status 2020, all parameters for all years must be specified for the regions valid for 2020. 

Luckily many official statistics institutions offer demographic data in which the information is given specifically for the most up-to-date area-status -- in retrospective and in forecasts. As a result, the current GEPOC ABM Geography version uses this strategy.

\subsubsection{Summary: Model Parameters}
\label{sec:parameters_geo}
We conclude the definition of this model extension by giving an update of the parameter tables introduced in Section \ref{sec:parameters}. Again, we do not specify how the corresponding parameter values can be found. General parameters are unchanged compared to Table \ref{tbl:params1}, demographic parameters are found in Table \ref{tbl:params2_geo}.

\begin{table}[H]
\caption{Demographic parameters of GEPOC ABM Geography. See Definition \ref{def:sex} for interpretation of the \textit{sex} variable.\vspace{2mm}}
\label{tbl:params2_geo}
\begin{center}
\begin{tabular}{p{1.7cm}|p{5.7cm}|c|c|p{4cm}}
\hline
Parameter & Dimensions & Unit & P. Space   & Interpretation\\
\hline
$\alpha_m$ & - &  probability & $[0,1]$  & probability for male \pa at birth\\
\hline
$a_{max}$ & - & years & $\mathbb{N}/\{0\}$ & maximum age regarded in the parameters\\
\hline
$r_x$ & $x\in\{0,d,e,b,i,min\}$& name & various & \regionallevels used for initialisation, death, emigration, birth and immigration processes.\\
\hline
$A_j^{r_x}$ & $x\in\{0,d,e,b,i,min\}, j\in \{1,\dots,q_x\}$& $\{(long,lat)\}$ & $\subset \mathbb{R}^2$ & Specification of the \regionfamilies matching to the specified \regionallevels with a suitable area-status.\\
\hline
$P(y,i,s,a)$ & $y\in \{y_0,\dots,y_e\}$, $i\in \{1,\dots, q_0\}$, $\quad a\in \{0,\dots, a_{max}\}$, $s\in \{\text{male},\text{female}\}$ & persons & $\mathbb{N}\cup\{0\}$ &  total population per region $A_i^{r_0}$, age $a$, sex $s$ at the start of year $y$.\\
\hline
$I(y,i,s,a)$ & $y\in \{y_0,\dots,y_e\}$, $i\in \{1,\dots, q_i\}$, $\quad a\in \{0,\dots, a_{max}\}$, $s\in \{\text{male},\text{female}\}$ & persons & $\mathbb{N}\cup\{0\}$ &  total immigrants to region $A_i^{r_i}$ with age $a$ (at time of immigration), sex $s$ in the course of year $y$.\\
\hline
$D^p(y,i,s,a)$ & $y\in \{y_0,\dots,y_e\}$, $i\in \{1,\dots, q_d\}$, $\quad a\in \{0,\dots, a_{max}\}$, $s\in \{\text{male},\text{female}\}$ & probability & $[0,1]$ &  Probability of a person with sex $s$ living in region $A_i^{r_d}$, who has had its $a$-th birthday in year $y$, to die before its $a+1$-st birthday.\\
\hline
$E^p(y,i,s,a)$ & $y\in \{y_0,\dots,y_e\}$, $i\in \{1,\dots, q_e\}$, $\quad a\in \{0,\dots, a_{max}\}$, $s\in \{\text{male},\text{female}\}$ & probability & $[0,1]$ &  Probability of a person with sex $s$ living in region $A_i^{r_e}$, who has had its $a$-th birthday in year $y$, to emigrate before its $a+1$-st birthday.\\
\hline
$B^p(y,i,s,a)$ & $y\in \{y_0,\dots,y_e\}$, $i\in \{1,\dots, q_b\}$, $\quad a\in \{0,\dots, a_{max}\}$, $s\in \{\text{male},\text{female}\}$ & probability & $[0,1]$ &  Probability of a female person living in region $A_i^{r_b}$, who has had her $a$-th birthday in year $y$, to give birth to a child before her $a+1$-st birthday. This probability must compensate for multiple-births which are not depicted in the model.\\
\hline
\end{tabular}
\end{center}
\end{table}
It is not necessarily relevant to explicitly parametrise all \regionfamilies $(A_i^{r_x})_{i=1}^{q_x}$ for the \regionallevels $x\in\{d,e,b,min\}$ explicitly, e.g. via raster maps or borders coordinates, since we do not sample points inside them (in contrast to $x\in\{0,i\}$). It is sufficient to quantify the mapping functions $\phi(long,lat,r_{x})$. The latter can be  simplified drastically and might not even require additional input data, if the different regional levels can be ordered w.r. to fineness (compare Definition \ref{def:fineness}).

For example, consider Austrian \regionallevels $r_0=\text{municipalities}$ and $r_d=\text{districts}$. Since the first three digits of the five digit \regionid of a municipality region is precisely the \regionid of the district, we do not need any additional information to compute $\phi(long,lat,r_d)$ from $\phi(long,lat,r_0)$.

\newpage
\section{GEPOC ABM Internal Migration - Model Definition}
\label{sec:immodel}
For long-range simulations GEPOC ABM Geography will always cause deviations for the regional age distributions primarily due to missing countryside$\leftrightarrow$city migration. To overcome this problem, we developed GEPOC ABM Internal Migration, henceforth short GEPOC ABM IM, as an \important{extension of GEPOC ABM Geography, defined in Section \mbox{\ref{sec:geographymodel}}}. 

\important{In the following we will define not one but three different models for internal migration which differ in strategy for location-sampling and parametrisation:}
\begin{itemize}
    \item \important{\textbf{Interregional model,}}
    \item \important{\textbf{Biregional model,}}
    \item \important{\textbf{Full Regional model,}}
\end{itemize}
compare with \cite{obermair2018different}. We will explain the three models at once building and extending the existing blocks of the ODD protocol from Section \ref{sec:geographymodel}.

\subsection{Overview}
\subsubsection{Purpose and Patterns}
In contrast to GEPOC ABM Geography we may also use this model extension for any kind of research question related to long-term regional change of the population and to investigate problems specifically related or caused by internal migration. Note, that this model is not intended to replace GEPOC ABM Geography since it is computationally more costly, structurally more complex and requires more parametrisation data.

\subsubsection{Entities, State Variables and Scales}
There are no changes to entities, variables and scales compared to GEPOC ABM Geography.

\subsubsection{Process Overview and Scheduling}
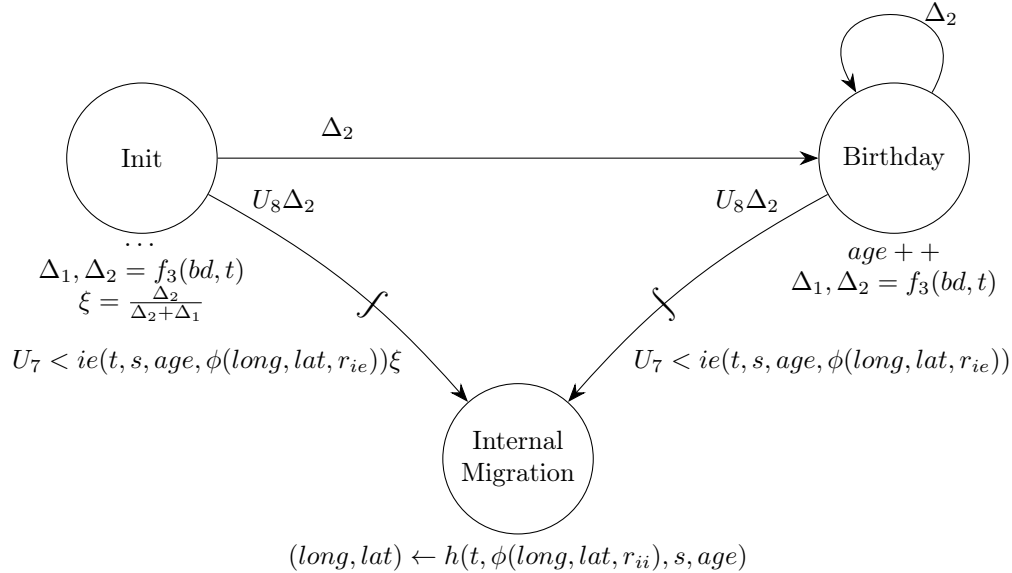
\begin{figure}[H]
    \centering
    \begin{tikzpicture}[
eventtype/.style = {circle, draw, minimum width = 2 cm,text=black,draw=black, align = center},
interface/.style = {rectangle, draw, minimum width = 2 cm,text=black, align = center},
scheduleplain/.style = {-{Stealth[scale=1.5]}},
nodeif/.style = {pos=0.66,anchor=center, circle, minimum width=0.5cm},
nodetime/.style = {pos=0.05,anchor=center},
nodearg/.style = {pos=0.4,anchor=center, rectangle, minimum width=0.5 cm,fill=white,draw},
]

\node[eventtype,label={[align=center]below:\dots\\$\Delta_1,\Delta_2=f_3(bd,t)$\\$\xi=\frac{\Delta_2}{\Delta_2+\Delta_1}$}] (init) at (0,0) {Init};

\node[eventtype,label={[align=center]below:$age++$\\$\Delta_1,\Delta_2=f_3(bd,t)$}] (bd) at (10,0) {Birthday};
\draw (bd) edge[scheduleplain,out = 60, in = 120, looseness = 4] node[nodetime,pos=0.2,label ={above:$\Delta_2$}]{} (bd);

\node[eventtype,label = {below:$(long,lat)\leftarrow h(t,\phi(long,lat,r_{ii}),s,age)$}] (im) at (5,-4) {Internal\\Migration};

\draw (bd) edge[scheduleplain,bend right = 10] node[nodetime,pos=0.2,label ={above:$U_8\Delta_2\quad~$}] {} node[nodeif,pos=0.85,label ={right:$U_7<ie(t,s,age,\phi(long,lat,r_{ie}))$}]{} node[pos=0.6,rotate=50] {\large{$\int$}} (im);

\draw (init) edge[scheduleplain,bend left = 10] node[nodetime,pos=0.2,label ={above:$\quad U_8\Delta_2$}] {} node[nodeif,pos=0.85,label ={left:$U_7<ie(t,s,age,\phi(long,lat,r_{ie}))\xi$}]{} node[pos=0.6,rotate=330] {\large{$\int$}} (im);

\draw (init) edge[scheduleplain] 
node[nodetime,pos=0.2,label={above:$\Delta_2$}]{}
(bd);
\end{tikzpicture}
    \caption{Internal-migration related snippet of the \personagent-layer of the time-update concept of GEPOC ABM IM using event-graph-like notation. Functions $f_i$ are explained in the text.}
    \label{fig:update_concept_pa_im}
\end{figure}
\paragraph{\Personagent layer.} While the simulation layer and the \interfaceagent layer remain entirely unchanged, the \pa layer is extended by one additional process: internal migration.

This is displayed in Figure \ref{fig:update_concept_pa_im}, which extends Figure \ref{fig:update_concept_pa} by the \textit{Internal Migration} event. 

Analogous to the \textit{Death}, \textit{Emigration}, and \textit{Birth} event it is scheduled randomly comparing a $U(0,1)$ random number with the value of a parameter function. Using the notation and definitions of Section \ref{sec:geography_entities}, let $r_{im}$ define the \regionallevel used for internal (e)migration with \regionfamily $(A_i^{r_{im}})_{i=1}^{q_{im}}$, then
\[ie(t,s,age,\phi(long,lat,r_{im}))\]
stands for the probability, that a person with sex $s$ living in region with \regionid $\phi(long,lat,r_{im})$, who turned $age$ at time $t$, moves due to internal migration before its $age+1$-st birthday. \important{We call this process ``internal emigration''.} Note that we do not specify that the person needs to leave the region. It is possible to internally migrate within the same region.

In case the agent is selected for internal migration, function $h$ samples a new residence, a process we usually call \important{``internal immigration''}. This is done in two steps: $h=h_{1}\circ g_{12}$. Function $h_{1}$ randomly samples a new region of residence for the agent given the current time its age, sex and current residence region. Thereafter, function $g_{12}$, which is the rejection-sampling algorithm introduced in Section \ref{sec:parameter_functions_geo}, draws a random new residence location in form of longitude and latitude.

\important{Dependent on the used model, $h_{1}$ differs:}
\begin{itemize}
    \item \textbf{Interregional model.} Sampling of a new region depends on time, region of origin and sex:
    \[h_{1}(t,\phi(long,lat,r_{im}),s,age):=h^{ir}_{1}(t,\phi(long,lat,r_{im}),s).\]
     \item \textbf{Biregional model.} Sampling of a new region depends on time, sex and age:
    \[h_{1}(t,\phi(long,lat,r_{im}),s,age)):=h^{br}_{1}(t,s,age).\]
    \item \textbf{Full Regional model.} Sampling of a new region depends on time, region of origin, sex and age:
    \[h_{1}(t,\phi(long,lat,r_{im}),s,age):=h^{full}_{1}(t,\phi(long,lat,r_{im}),s,age).\]
\end{itemize}

\subsection{Design Concepts}
\subsubsection{Basic Principles}
While the general logic of the IM extension of GEPOC ABM is easy to understand -- it simply adds one additional demographic process -- reasoning for development of three different models is required. There are two main motivations for this:
\begin{itemize}
    \item \important{availability of parametrisation data, and}
    \item \important{size of parametrisation data.}
\end{itemize}
Suppose one is keen to parametrise the Full Regional model, one needs to setup a probability distribution for all possible regions a \pa can move into, i.e. $q_{im}$, for any given region of origin, sex, age and simulation time/simulation year. We give an example for the sheer amount of data required this way: Say $r_{im}$ is set to Austrian municipalities with $q_{im}\approx 2000$, we choose $a_{max}=100$ and aim to have a stable parametrisation for $50$ years, then we would require
\[q_{im}\cdot q_{im}\cdot |\{0,\dots,a_{max}\}|\cdot |\{\text{male},\text{female}\}|\cdot |\{y_0,\dots,y_{49}\}|\approx 2000\cdot 2000\cdot 101\cdot 2 \cdot 50 = 4.04\cdot 10^{10}\]
data points to fully parametrise the model. Collecting $40$ billion valid data points is not only a huge task for the parametrisation but also for keeping the data in the memory.

As a workaround, we may leave out one of the costly dimensions: The Interregional model leaves out the age variable, which reduces the data requirement by one hundredth, the Biregional model neglects the dependency of the origin region, which, on the given example, reduces the requirement by one two-thousandth. 

So far we were not been able to successfully parametrise and compute the Full Regional model, but the two reduced models. The model simplification, of course, comes with a price regards validity: While the Interregional model perfectly depicts the flows between the different regions, it does not correctly depict the age structure which will lead to demographic age-deviations in the long run. In the contrast, the Biregional model perfectly depicts the age structure of the internal immigrants, but does not correctly model the flows between the regions. This will lead to correct demographic development based on potentially wrong migration processes. Therefore the user has to decide, which version of validity is more relevant for the specific use case.

\subsection{Details}
\subsubsection{Submodels}
\label{sec:parameter_functions_im}
We furthermore explain the used parameter functions in detail and connect with the parameterisation.
\paragraph{Function $ie$.}
This parameter function is defined analogously to all other probabilities in Section \ref{sec:parameter_functions_geo}:
\begin{equation}
    \left(\begin{array}{c}d_{new}(t,i,s,a)\\
    e_{new}(t,i,s,a)\\b_{new}(t,i,a)\\ie_{new}(t,i,s,a)\end{array}\right)=\Psi\left(\begin{array}{c}D^p(y,i,s,\min(a,a_{max}))\\E^p(y,i,s,\min(a,a_{max}))\\
    B^p(y,i,\min(a,a_{max})\\IE^p(y,i,s,\min(a,a_{max}))\end{array}\right)
\end{equation}
Hereby, $IE^p(y,i,s,a)$ stands for the probability that a person with sex $s$, living in region $A_i^{r_{im}}$, and aged $a$ in the course of year $y$, emigrates internally before the person's $a+1$-st birthday, and function $\Psi$ corrects the bias due to simultaneous events (see Theorem \ref{thm:aposterior_aprior}).

\paragraph{Functions $h_1^{ir},h_1^{br}$ and $h_1^{full}$.}
We furthermore define:
\begin{align}
M^p(y,i,s,a,j)\quad\dots & \text{  Pr. of an emigrant from $i$ (sex $s$, age $a$) to migrate to $j$ during $y$.}\\
II^p(y,j,s,a):=\sum_{i}M(y,i,s,a,j)\quad\dots & \text{   Pr. of an emigrant (sex $s$, age $a$) to migrate to $j$ during $y$,}\\
OD^p(y,i,s,j):=\sum_{a}M(y,i,s,a,j)\quad\dots & \text{   Pr. of an emigrant from $i$ (sex $s$) to migrate to $j$ during $y$.}
\end{align}
Hereby, $II^p$ can be regarded as immigration probability into a certain region, and $OD$ can be interpreted as origin-destination flow probability between the regions.

As usual, let $t=(y,m,d)$, then
\begin{align}
h^{ir}_{1}(t,i,s)=X&\text{   with   } Pr(X=j|t,i,s)=OD^p(y,i,s,j),\\
h^{br}_{1}(t,s,a)=X&\text{   with   } Pr(X=j|t,s,a)=II^p(y,j,s,min(a,a_{max})),\\
h^{full}_{1}(t,i,s,a)=X&\text{   with   } Pr(X=j|t,i,s,a)=M^p(y,i,s,min(a,a_{max}),j).
\end{align}

\subsubsection{Mixing Strategies}
The three models introduced clearly open the ideas to be mixed, in particular when different \textit{regional-levels} are used. For example, one may use a fine regional-level for sampling internal emigration, a coarse regional-level to sample a rough immigration region with the Full Regional model, and, again, a fine regional-level with the Biregional model to refine the sampled region. 

In case the Full Regional model is out of scope w.r. to gathering parametrisation data, we may also investigate
\begin{equation}
    ii(t,i,s,a)=X\text{   with   }Pr(X=j|t,i,s,a)=F(II^p(y,j,s,a),OD^p(y,i,s,j))
\end{equation}
with some function $F$ combining the two probabilities, e.g. via a linear combination. Unfortunately, the minimalist example in Appendix \ref{sec:minimodel} shows that it is not so simple. Both models, the Biregional and the Interregional, fulfil a certain perspective of validity. The former conserves the correct regional age-distributions, the latter conserves the correct migration flows. Unfortunately, using a simple function $F$ such as a multiplication or linear combination, neither of the two constraints can be conserved.

Generalising the equations derived in Appendix \ref{sec:minimodel} we find the necessary requirements for function $F$ to conserve both constraints. Let $II^p$ stand for the internal immigration probabilities (with age resolution) and $OD^p$ stand for the origin destination probabilities (without age resolution), then both can be joined to create probabilities $\tilde{M}^p(y,i,s,a,j)$ which maintain both constraints by solving a series of under-determined linear problems: For every required year $y$ and sex $s$, find $\tilde{M}^p: 0\leq \tilde{M}^p\leq 1$ so that
\begin{equation}
    \forall j\in\{1,\dots,q_{im}\},a\in \{0,\dots,a_{max}\}:\sum_{i=1}^{q_{im}}\frac{IE(y,s,i,a)}{\sum_{k=1}^{q_{im}} IE(y,k,s,a)}\cdot \tilde{M}^p(y,i,s,a,j)=II^p(j,a),
\end{equation}
\begin{equation}
    \forall i,j\in \{1,\dots,q_{im}\}:\sum_{a=0}^{a_{max}}\frac{IE(y,i,s,a)}{\sum_{b=0}^{a_{max}} IE(y,i,s,b)}\cdot \tilde{M}^p(y,i,s,a,j)=OD^p(y,i,s,j).
\end{equation}
\begin{equation}
    i\in\{1,\dots,q_{im}\},a\in \{0,\dots,a_{max}\}:\sum_{j=1}^{q_{im}}\tilde{M}^p(y,i,s,a,j)=1.
\end{equation}
The corresponding problem has $q_{m}^2\cdot (a_{max}+1)$ degrees of freedom and $2q_{m}\cdot (a_{max}+1)+q_{m}^2$ equations. As seen on the minimalist example in Section \ref{sec:minimodel}, the problem is heavily under-determined and large: With the aformentioned example for Austria, i.e. using $q_{im}\approx 2000$ municipalities and $a_{max}=100$, we would need to find $2000^2\cdot 101=404M$ parameter values based on $2\cdot 2000\cdot 101+2000^2=4 404 000$ constraint equations.

This computation is clearly not suitable for simulation run-time since it is by itself a huge challenge for even the most powerful linear programming solvers, but it could be used in a pre-processing step to find a plausible parametrisation for the Full Regional model. A heuristics was already able find a solution to the problem on the district level in Austria with around $1$ Million parameter values and $10000$ equations.

\subsubsection{Summary: Model Parameters}
\label{sec:parameters_im}
We conclude the specification of this model extension by giving an update of the parameter tables introduced in Section \ref{sec:parameters_geo}. Again, we do not specify how the corresponding parameter values can be found. General parameters are unchanged compared to Table \ref{tbl:params1}, general demographic parameters without internal migration are found in Table \ref{tbl:params2_geo}, internal migration parameters are found in Table \ref{tbl:params2_im}.

\begin{table}[H]
\begin{tabular}{p{1.9cm}|p{5.7cm}|c|c|p{4cm}}
\hline
Parameter & Dimensions & Unit & P. Space   & Interpretation\\
\hline
$r_{im}$ & - & name & various & regional-level used for internal migration.\\
\hline
$A_j^{r_{im}}$ & $j\in \{1,\dots,q_{im}\}$& $\{(long,lat)\}$ & $\subset \mathbb{R}^2$ & Specification of the regional set-families for internal migration.\\
\hline
$IE(y,i,s,a)$ & $y\in \{y_0,\dots,y_e\}$, $i\in \{1,\dots, q_0\}$, $\quad a\in \{0,\dots, a_{max}\}$, $s\in \{\text{male},\text{female}\}$ & probability & $[0,1]$ &  Probability of a person with sex $s$ living in region $i$, who has had its $a$-th birthday in year $y$, to emigrate internally before its $a+1$-st birthday.\\
\hline
\multicolumn{5}{c}{Interregional model}\\
\hline
$OD(y,i,s,j)$ & $y\in \{y_0,\dots,y_e\}$, $i\in \{1,\dots, q_{im}\}$, $s\in \{\text{male},\text{female}\}$, $j\in \{1,\dots, q_{im}\}$ & persons & $\mathbb{N}\cup\{0\}$ &  total migrants from region $i$ to $j$ with sex $s$ in the course of year $y$.\\
\hline
\multicolumn{5}{c}{Biregional model}\\
\hline
$II(y,j,s,a)$ & $y\in \{y_0,\dots,y_e\}$, $j\in \{1,\dots, q_{im}\}$, $s\in \{\text{male},\text{female}\}$, $a\in \{0,\dots, a_{max}\}$ & persons & $\mathbb{N}\cup\{0\}$ &  internal immigrants into region $j$ with sex $s$ and age $a$ in the course of year $y$.\\
\hline
\multicolumn{5}{c}{Full Regional model}\\
\hline
$M(y,i,s,a,j)$ & $y\in \{y_0,\dots,y_e\}$, $i\in \{1,\dots, q_{im}\}$, $s\in \{\text{male},\text{female}\}$, $a\in \{0,\dots, a_{max}\}$, $j\in \{1,\dots, q_{im}\}$ & persons & $\mathbb{N}\cup\{0\}$ &  internal migrants from region $i$ into $j$ with sex $s$ and age $a$ in the course of year $y$.\\
\hline
\end{tabular}
    \centering
    \caption{Internal migration parameters of GEPOC ABM IM dependent of the use migration model. See Definition \ref{def:sex} for interpretation of the \textit{sex} variable.}
    \label{tbl:params2_im}
\end{table}
In contrast to regional-levels for e.g. birth or emigration it is relevant to explicitly parametrise the $(A_i^{r_{im}})_{i=1}^{q_{r_{im}}}$ since we need to sample points inside the regions them.

\newpage
\section{GEPOC ABM Contact Location - Model Definition}
\label{sec:contactmodel}
\important{GEPOC ABM Contact Location, henceforth GEPOC ABM CL, extends GEPOC ABM Geography, see Section \mbox{\ref{sec:geographymodel}}, by features related to agent-agent contacts.}

\subsection{Overview}
\subsubsection{Purpose and Patterns}
Key purpose of this model extension is to be a foundation for models relying on in-person contacts between {\pa}s, for example, epidemiological models. Note that the model itself does not sample any contacts but provides an underlying contact network as a basis for them. We will give some hints on how to model contacts using the defined network in Section \ref{sec:contactmodel_interaction}.

\subsubsection{Entities, State Variables and Scales}

In addition to the two agent types \pa and \interfaceagent introduced earlier, \important{we add two new passive agent types to the model:} \footnote{We are aware, that many modellers would not consider passive entities as ``agents''. Since we think of potential model extensions, in which the entities might gain active roles, we stick with this notation though.}
\begin{itemize}
    \item \important{\mbox{\location}, and}
    \item \important{\mbox{\locationcollection}.}
\end{itemize}
The former models a place where {\pa}s meet, the latter models a place which summarises {\location}s and works as a platform for additional {\pa} contacts in between the different summarised {\location}s. Typical examples of {\location}s are households, school-classes, workplaces, whereas classic examples for {\locationcollection}s are schools, company-buildings or care-homes. \important{As hinted by these examples, it is possible to use multiple different sub-types of \mbox{\location} or \mbox{\locationcollection} in the model. They might also come with different features and constraints. Nevertheless, the base sampling-mechanism for the network and the general parametrisation concept is the same for all of them.}

The \location agent has four states, namely
\begin{itemize}
    \item a set $P_{loc}$ of {\pa}s assigned to the \location,
    \item $long$, the longitude of the {\location}'s position, and
    \item $lat$, the latitude of the {\location}'s position, and
    \item $\vec{c}\in(\mathbb{N}\cup \{0\})^K$, referring to a {\location}'s vector of initial \pa capacities for agents with respect to $K$ different predefined characteristics $k$ with mappings $\Lambda_k,k\in\{1,\dots,K\}$ (compare Definition \ref{def:characteristic}).
\end{itemize}
The latter is solely used within the initialisation process of the model and specifies how many agents with which characteristic are scheduled for the specific location. After successful initialisation (see below) we would have $(\vec{c})_k=\sum_{a\in P_{loc}}\Lambda^k(a(t_0))$.
\newpage
\example{Household initialisation with characteristics.}{households}{The idea behind the usage of characteristics in this process is best explained with an example: Suppose the \location agents are used to model households, we might require to depict a correct age and sex distribution of the households as given by the census data. We might introduce 
\begin{align}
\Lambda^{1}(a(t)):=1_{age<18}(a(t))\\
\Lambda^{2}(a(t)):=1_{18\leq age<65,sex=f}(a(t)),\quad &\Lambda^{3}(a(t)):=1_{18\leq age<65,sex=m}(a(t))\\
\Lambda^{4}(a(t)):=1_{65\leq age,sex=f}(a(t)),\quad &\Lambda^{5}(a(t)):=1_{65\leq age,sex=m}(a(t)).
\end{align}
With these characteristics defined, according to Austrian data (Statistics Austria), about $30\%$ of all households should consist of one adult male and female, i.e. $\vec{c}=(0,1,1,0,0)^T$, around $13\%$ of two opposing sex elderly, i.e.  $\vec{c}=(0,0,0,1,1)^T$, and about $12\%$ of one adult male alone, i.e. $\vec{c}=e_2$. Furthermore, about $15\%$ or all households have children and two opposite sex parents, i.e.$\vec{c}=(\geq 1,1,1,0,0)^T$. (See Definition \ref{def:sex} for interpretation of the \textit{sex} variable)}

The \locationcollection is a special type of \location and has the following four features:
\begin{itemize}
    \item a set of {\location}s assigned to the \locationcollection,
     \item $lat$, the latitude of the {\locationcollection}'s position, and
      \item $long$, the longitude of the {\locationcollection}'s position,
     \item $\vec{c}\in (\mathbb{N}\cup \{0\})^J$, referring to the {\locationcollection}'s vector of initial \location agent capacities with respect to $J$ different \location agent types capable for being assigned to the \locationcollection type.
\end{itemize}
Since the \locationcollection is regarded as a special type of \location agent, it is possible to define a \locationcollection agent which is assigned a set of other \locationcollection agents.

The last feature is analogous to the capacity feature of the \location agent but refers to (potentially) multiple types of \location agents.
\example{School as location-collection.}{school}{As before, the idea is best explained with an example: Suppose the \locationcollection agents are used to model schools, then the included \location agents could model school-classes consisting of pupils but also one or more workplaces for teachers. Correspondingly the capacity vector may be two dimensional and e.g. $\vec{c}=(20,1)^T$ states that the \locationcollection should contain 20 school-class {\location}s and one workplace \location. The pupil {\pa}s and the teacher {\pa}s may all interact with each other beyond their own school-class and work-place \location via the \locationcollection environment on a less frequent basis.
}

\important{In this specification of GEPOC ABM CL we do not regard any down-or up-scaling of the model. I.e. for this original specification $\sigma=1$.} The reason for this lies in intrinsic problems to scale contact networks in general. Suppose, with $\sigma=1.0$ there are $50$ school-classes with $20$ pupils. How many classes and pupils-per-class are correct for $\sigma=10$ to conserve the network features? Neither $50/2$ nor $5/20$ (nor anything in between) would provide qualitatively equivalent results as the original version e.g. if the locations were used for contacts in an epidemics model.

\subsubsection{Process Overview and Scheduling}
\label{sec:contacts_scheduling}
\textbf{Simulation layer.} The most relevant changes for the simulation dynamics are found on the simulation layer and, in specific, in the initialisation loop (compare Figure \ref{fig:update_concept_sim}). In this particular case, having an Event Graph representation would not be helpful, since the dynamics of the model are, in principle, unchanged, and the additions to the initialisation process are too complicated to be described in this fashion. Thus we describe the changed dynamics in textual form in temporal order, i.e. in the order/sequence in which they are executed by the simulation. Note, that they all take place as a part of the initialisation process at $t=0$ and in between the initialisation of the \pa population and the first \textit{Observe event} (see Section \ref{sec:basemodel} for details). 

\important{Generating of \mbox{\location} and \mbox{\locationcollection} agents starts directly after the \mbox{\pa} population is generated. Clearly, assignment of a certain type of \mbox{\locationcollection} agents can only be started if the member \mbox{\location}/ \mbox{\locationcollection} populations have been generated before.} The last restriction creates a natural hierarchy on the different types of \mbox{\location} and \mbox{\locationcollection} agents.

\important{Generation of each \mbox{\location}/\mbox{\locationcollection} population is a two step process. In the first step we initialise a certain number of \mbox{\location}/\mbox{\locationcollection} agents, and sample residence coordinates and initial capacities. In a second step, the corresponding \mbox{\pa}/\mbox{\location} agents are assigned to them in a ``filling'' process.} For the first part, we require, again, a \regionallevel and a corresponding \regionfamily to create and distribute the initial locations. For the second part, we require origin-destination information to setup the regional assignment of agents. For details we refer to Section \ref{sec:contactmodel_initialisation}.

\textbf{\Personagent layer.} The concept of contact {\location}s per-se does not influence the \pa dynamics, yet it might be required to change \location assignments on run-time due to agent behaviour. \important{First of all, if the \mbox{\pa} emigrates or dies, it has to be removed from all \mbox{\location} agents which it is assigned to in the course of the \mbox{\textit{Remove}} event} (see Figure \ref{fig:update_concept_pa}). This might result in new \location agents with empty \pa-set. So the modeller should be careful when attempting to draw {\pa}s from arbitrary locations. \important{Furthermore, immigrated and newborn \mbox{\pa} must be added to \mbox{\location} agents in the course of the {\textit{Add}} event} (see Figure \ref{fig:update_concept_sim}). We do not specify how this should be done since the process depends on the specific contact place and purpose depicted by the \location type. \important{Dependent on the specific application, also other events might require to change the network affiliation of agents.} Examples are the {\textit{Internal Migration}} event introduced in GEPOC ABM IM which causes the agent to find a new place to live. Also the \textit{Birthday Event} might trigger a change of the network since it causes the agent to age by one year.

\subsection{Design Concepts}
\subsubsection{Basic Principles}
There are two main ideas behind the concept of creating contact networks via locations.

The first one is the motivation to model and investigate human behaviour within different settings whereas we are also capable of interfering with the setting itself. E.g. workplace-{\location}s could be set \textit{closed} due to enforced home-office or lock-downs (compare with \cite{Bicher_Rippinger_Urach_Brunmeir_Siebert_Popper_2021}). The second idea is motivated by parametrisation considerations. Typically, human contacts in agent-based models are modelled using a random scale-free network such as the Barabasi-Albert Graph\cite{albert2002statistical}. Unfortunately, it is very difficult (if not impossible) to find parameters for designing and parametrising a scale-free spatial network which depicts heterogeneous regional features such as high/low inter-connectivity (good/bad public transport) or population structure (e.g. age-structure). Using the proposed approach, the modeller can make use of statistical census data (e.g. labour, school, household statistics, etc) to gather the required information, i.e. the regional number of locations, the characteristics of the individuals within, and the origin-destination map for assigning the {\pa}s. Anyway, if parameter values are selected/collected properly, the network will show features of a small-world/scale-free network, in particular, if the modeller follows the contact generation concepts described below in Section \ref{sec:contactmodel_interaction}.

\subsubsection{Interaction}
\label{sec:contactmodel_interaction}
The described model extension provides a proper basis for modelling human-human interaction in GEPOC ABM, yet does not specify how contacts are actually handled. In this section, we will summarise the most important concepts from \cite{Bicher_Rippinger_Urach_Brunmeir_Siebert_Popper_2021} to give the reader some ideas how contacts can be modelled based on the given location-based network.
\paragraph{Gamma-Poisson Mix.} Skewness and high clustering is one of the most important features of a realistic human contact network. Since the underlying network via locations itself does not include a mechanism for adding heterogeneity within the contact behaviour of the individual {\pa}s, skewness/dispersion can be not originate from the network of potential contact partners, i.e. the number and members of the locations assigned to the agent, but solely from the daily number of drawn contacts on model run-time (e.g. for spreading a disease it is not relevant how many people you known, but how many and how often you meet them). In \cite{Bicher_Rippinger_Urach_Brunmeir_Siebert_Popper_2021} this heterogeneity is modeled via a scalar \textit{contactivity} $c$ parameter as additional \pa parameter, which models the agents personal appeal to have many contacts. In the study the parameter value is initialised randomly in the initialisation process of the \pa by sampling a gamma distributed random variable with mean $1$. The value of the second free parameter of the gamma distribution is calibrated to a measured dispersion factor from a published study (\cite{Adam_Wu_Wong_Lau_Tsang_Cauchemez_Leung_Cowling_2020}). Furthermore, on runtime, a Poisson distribution is used to sample the actual contacts per time-step. Given the average number of contacts $n$ per time-step within a specific location type from a parameter file, the model would draw $[Poi(c\cdot n)]$ not necessarily distinct contact partners from the \pa set of the assigned \location agent. The agent would furthermore generate contact events with every one of them.

\paragraph{Contact Events.} Since it is highly recommended that every agent can only change its own states, we advise to add \textit{Contact Events} into the event-queues of all sampled contact partners. Since the correct state of both agents can only be ensured at the times, the overall model is in sync, contacts can only be planned and executed at the discrete point in time $t_i$. Therefore the event should always be scheduled in the course of the \textit{Time-Step Planning} events of the agents and added to the event queues of the contact partners without any additional delay.

\paragraph{Contacts within {\locationcollection}s.} To specify, how many contacts take place in between {\location}s inside of a \locationcollection agent, \cite{Bicher_Rippinger_Urach_Brunmeir_Siebert_Popper_2021} used a scalar probability parameter within the \locationcollection agent. Any contact drawn within any of the {\location}s summarised in the \locationcollection is instead drawn from the joint set of all {\pa}s of all summarised {\location}s with the specified probability.

\subsection{Details}

\subsubsection{Initialisation}
\label{sec:contactmodel_initialisation}
In this section we take a deeper look into the initialisation of the \location and \locationcollection agents and explain the mentioned two-step process:
\paragraph{Initialisation of \location agents.} Analogous to the {\pa}s, we initialise the \location agents using a specific \regionallevel $r_l$ with \regionfamily $(A^{r_l})_{i=1}^{q_{l}}$ and a corresponding parameter-vector $R(i),i=1\dots,q_l$, which contains the total number of contact locations within region $A^{r_l}_i$:
    \begin{equation}
        i\in 1,\dots, q_l:\quad R_i=|\{\text{\location}\in A^{r_l}_i\}|.
    \end{equation}
\begin{itemize}
    \item We iterate over all $q_l$ regions and accordingly create the proper amount of \location agents. For each of them, we sample random residence coordinates inside the region. If the spatial distribution of the locations is equivalent or very similar to the distribution of persons (which depends on the used location type and/or the modelling purpose) this can be done analogously to the \pa as defined in Section \ref{sec:parameter_functions_geo}, i.e. using an even finer settlement map.
    \item In addition to the initial coordinate, we sample an initial capacity. With the given $K$ characteristics we assign $\vec{c}$ by drawing from the discrete distribution $Pr(\vec{c}=\vec{X}),\quad \vec{X}\in (\mathbb{N}\cup\{0\})^K$. Note, that for feasible parameter values, this distribution will always have a finite support and it should be possible to parametrise it properly with data.
\end{itemize}
\paragraph{Filling of \location agents.} Key for assigning \pa to {\location}s is an origin destination map $OD$, a static analogue to the one presented in GEPOC ABM IM (Section \ref{sec:immodel}). It matches the \regionallevel $r_l$ and states, how many individuals from each region (origin) are assigned to locations in each other region (destination). With
    \[\forall i\in \{1,\dots,q_l\}:\quad Pr(X=i|j)=\frac{OD(i,j)}{\sum_{i'=1}^{q_l}OD(i',j)}\]
    we get a discrete distribution for the origin region of a \pa who is to be assigned to a \location in region $A^{r_l}_j$.
\begin{itemize}
    \item In the first step, we create a map of unassigned {\pa}s and put them into bins according to their regional identifier and their characteristics:
    \[\forall i\in \{1,\dots,q_l\}, k\in \{1,\dots,K\}: G(i,k):=\{\pa:\phi(long,lat,r_{l})=i wedge \Lambda^k(\pa)=1\}.\]
    We will use this map to draw {\pa}s according to sampled origin region and characteristic.
    \item Furthermore, we iterate over all created {\location}s and over all characteristics $k\in\{1,\dots,K\}$. If the planned capacity $c_k$ of the \location agent is not zero, we try to assign accordingly many {\pa}s using the following system:
    \begin{enumerate}
        \item Let $j$ stand for the region identifier of the \location agent and investigate the set $I:=\{i\in\{1,\dots,q_r\}:OD(i,j)>0\}$ of all \regionids that specify a potential origin region for $j$ with positive probability. If 
        \[\bigcup_{i\in I}G(i,j)=\emptyset\]
        we will not be successful in finding a \pa with the required characteristic from a potential origin region. Therefore, we break the loop and continue with the next characteristic/\location-agent. Otherwise we continue with step 2.
        \item Draw a random origin region $i$ using the specified discrete distribution $Pr(X=i|j)$. If $G(i,k)=\emptyset$ continue with 2, otherwise continue with 3. Note that step 1 ensures that the algorithm will eventually find an origin region $i$ with non-empty set $G(i,k)$.
        \item Pick and remove a random \pa from $G(i,k)$ and assign the agent to the \location agent.
    \end{enumerate}
\end{itemize}
Note, that the mapping $G$ must be recreated for every new \location type since {\pa}s can, of course, be assigned to multiple contact locations at once. 
\paragraph{Initialisation and filling of \locationcollection agents.}
Creation and filling of \locationcollection agent works analogous to the one of \location agents. Instead of looping over the characteristics, a loop over the suitable \location agent types is performed.

\important{This initialisation strategy might lead to under-full or even entirely empty \mbox{\location}s or \mbox{\locationcollection}s (due to step 1 in the filling process) if parametrisation or source data for parametrisation is simplified or flawed.} Since this might cause problems, we recommend to remove entirely empty {\location}s and {\locationcollection}s from the model before continuing with the initialisation of the next type or starting with the model dynamics.

\subsubsection{Summary: Model Parameters}
\label{sec:parameters_cl}
We conclude the specification of this model extension by showing which additional parameter values are needed to create a contact network based on \location and \locationcollection agents. For every type of \location and \locationcollection agent we require the parameters shown in Table \ref{tbl:params1_cl}. As for the other models, we do not specify how the corresponding parameter values can be found.

\begin{table}[H]
\caption{Parameters required for every \location and \locationcollection agent type in GEPOC ABM CL.}
    \label{tbl:params1_cl}
    \begin{center}
\begin{tabular}{p{2.3cm}|p{5cm}|c|c|p{4.0cm}}
\hline
Parameter & Dimensions & Unit & P. Space   & Interpretation\\
\hline
$r_{l}$ & - & name & various & regional-level used for contact location.\\
\hline
$A_j^{r_{l}}$ & $j\in \{1,\dots,q_{l}\}$& $\{(long,lat)\}$ & $\subset \mathbb{R}^2$ & Specification of the regional set-families for random sampling of the contact location.\\
\hline
$R_j$ & $j\in \{1,\dots,q_{l}\}$ & locations & $\mathbb{N}\cup\{0\}$ &  Number of locations in region $A_j$.\\
\hline
$OD(i,j)$ & $i,j\in \{1,\dots,q_{l}\}$ & persons & $\mathbb{N}\cup\{0\}$ &  Number of \pa/\location agents in region $A_i$ assigned to a \location/\locationcollection in region $A_j$.\\
\hline
\multicolumn{5}{c}{\location only}\\
\hline
$K$ & - & number & $\mathbb{N}$ & A number of characteristics to distinguish when assigning {\pa}s to the \location agent.\\
\hline
$\Lambda^k$ & $k\in\{1,\dots,K\}$ & $\pa(t) \mapsto \Lambda^k(\pa(t))$ & $S \rightarrow \{0,1\}$ & Set of $K$ characteristic mappings to distinguish if an \pa has the characteristic (1) or not (0).\\
\hline
$Pr(\vec{c}=\vec{X}|j)$ & $r\in \{1,\dots,q_l\},\vec{X}\in (\mathbb{N}\cup\{0\})^K$ & probability & $[0,1]$ & Discrete distribution, how many agents with which characteristics are planned to be assigned to a location. The spatial resolution $j$ is optional.\\
\hline
\multicolumn{5}{c}{\locationcollection only}\\
\hline
$J$ & - & number & $\mathbb{N}$ & A number of \location agent types which should be assigned to the specific \locationcollection type.\\
\hline
$Pr(\vec{c}=\vec{X}|j)$ & $j\in \{1,\dots,q_l\},\vec{X}\in (\mathbb{N}\cup\{0\})^J$ & probability & $[0,1]$ & Discrete distribution, how many {\location}s from which of the $J$ location types are assigned to the \locationcollection. The spatial resolution $j$ is optional.\\
\hline
\end{tabular}
\end{center}
\end{table}

\newpage
\section{A-Posterior to A-Prior Probabilities}
\label{sec:aposterior_aprior}
A given probability $X^p$ (with $X\in \{E,B,D\}$) from a census bureau would, in principle, be well suited to be used as a probability in GEPOC ABM if the model would only regard one single mechanism (e.g. birth, death \underline{or} emigration). Yet the simultaneous presence of all three mechanisms causes a bias. 

\subsection{Motivation}
To make this problem clear, we introduce two models for the same system: In both cases the model returns a number between $0$ and $N>1$.
\begin{model}[Model 1]
Let $P_i,i\in \{1,\dots,N\}$ stand for the so called \textit{a-posterior} probability that an event with type $i$ occurs and let $P=\sum_{i=1}^N P_i$ be the overall probability of an event which we assume to be smaller than one. First, a Bernoulli experiment draws a random number $X$ which is equal to $1$ with probability $P$. In case $X=1$, an element $Y$ from $\{1,\dots,N\}$ is drawn with the discrete distribution $P(Y=i)=\frac{P_i}{P}$ and returned, otherwise the model returns $0$.
\end{model}

\begin{model}[Model 2]
Let $p_1,i\in \{1,\dots,N\}$ be \textit{a-prior} probabilities and sample $N$ random numbers $(X_i)_{i=1}^{N}$ with values in $\{0,1\}$ whereas $P(X_i=1)=p_i$. Furthermore, define the index set $I=\{i\in \{1,\dots,N\}:X_i=1\}$. If $I\neq \emptyset$, then a random index $i$ of $I$ is picked and returned, otherwise the model returns zero.
\end{model}

Model 1 uses a very natural parametrisation since the probabilities can be calculated from observations, since $P(\text{Model 1}=i)=P_i$. Therefore the output-probability matches the given input probability. This is not the case for Model 2 since we need to investigate the conflicting co-scheduling of any two events. 

In case we aim that both models lead to the same results, the following Corollary holds for $N=2$.
\begin{corollary}[A-Prior vs. A-Posterior ($N=2$)]\label{aprior_aposterior_n2}
We find the relations
\[P_1 = p_1(1-\frac{1}{2}p_2),\]
\[P_2 = p_2(1-\frac{1}{2}p_1),\]
and
\[p_1=1+\frac{P_1-P_2}{2}-\sqrt{1-P+\frac{(P_1-P_2)^2}{4}},\]
\[p_2 =1+\frac{P_2-P_1}{2}-\sqrt{1-P+\frac{(P_1-P_2)^2}{4}},\]
to guarantee that, in probability, Model 1 and Model 2 give the same results.
\end{corollary}
\begin{proof}
    We find that Model 2 returns 1 in precisely two cases: (1) the first Bernoulli experiment returns true while the second does not, or (2) both experiments return true and index 1 is chosen randomly from the set $\{1,2\}$ - which is a fair coin flip with chance $1/2$. The probability writes to
\[P_1=p_1(1-p_2)+1/2p_1p_2=p_1(1-\frac{1}{2}p_2),\]
analogous, $P_2$:
\[P_2=p_2(1-p_1)+1/2p_1p_2=p_2(1-\frac{1}{2}p_1).\]
Subtraction of the two equations leads
\[P_1-P_2=p_1-p_2\Rightarrow p_2=P_2-P_1+p_1.\]
Combining in combination with the first equation, we get
\[P_1=p_1(1-\frac{1}{2}(P_2-P_1+p_1))\Rightarrow p_1^2+p_1(P_2-P_1-2)+2P_1.\]
Solving the quadratic equation gives
\[\Rightarrow (p_1)_{1,2}=\frac{2+P_1-P_2}{2}\pm \sqrt{\frac{(2+P_1-P_2)^2}{4}-2P_1}.\]
Expanding the quadratic term and using $P_1+P_2=P$ the formula simplifies to
\[
\Rightarrow (p_1)_{1,2}=1+\frac{P_1-P_2}{2}\pm \sqrt{1-P+\frac{(P_1-P_2)^2}{4}.}
\]
Only the solution with "$-$" makes sense here: If $P_1>P_2$, then $1+\frac{P_1-P_2}{2}>1$ and adding the value of the root would make it even greater. This violates the condition for $p_1$ being a probability (i.e. $0\leq p_1\leq 1$). Otherwise, $1+\frac{P_1-P_2}{2}<1$, yet the value of the root is always greater than $|\frac{P_1-P_2}{2}|$ and adding it would also cause $p_1>1$.
\end{proof}
Considering the seemingly simple initial situation, the found solution is surprisingly complex. Hence, it is not surprising, that, so far, no analytic formula for $N>2$ could be found.

\subsection{Application in GEPOC ABM}
The subsystem of GEPOC ABM consisting of emigration and deaths is precisely like Model 2 with $N=2$: For both events a random process decides if the event will be scheduled in the course of a person-agents upcoming life-year. In case both are scheduled simultaneously, it is eventually a coin-flip, which of the two is scheduled earlier and will take place. The other one is cancelled since the agent is removed.
\newpage

Unfortunately, this does not only affect the death and emigration probabilities. Although the other events occurring in GEPOC ABM do not interfere with the death and emigration processes, they are implicitly influenced by them. We summarise the correct parameter-post-processing in the following theorem:

\begin{theorem}[A-Posterior to A-Prior (GEPOC)]\label{thm:aposterior_aprior}
With given a-posterior probabilities $D^p,E^p$ for emigration and death, and additional non-terminal probabilities $X_1^p,X_2^p,\dots,X_n^p$, e.g. for birth and internal migration, we define
\[\Psi:[0,1]^{n+2}\rightarrow [0,1]^{n+2}\] via
\begin{equation}\Psi_1(D^p,E^p,X_1^p,\dots)=1+\frac{D^p-E^p}{2}-\sqrt{1-(E^p+D^p)+\frac{(D^p-E^p)^2}{4}},\end{equation}
\begin{equation}\Psi_2(D^p,E^p,X_1^p,\dots) =1+\frac{E^p-D^p}{2}-\sqrt{1-(E^p+D^p)+\frac{(D^p-E^p)^2}{4}},\end{equation}
and $\forall 2<i\leq (n+2)$
\begin{equation}
\Psi_i(D^p,E^p,X_1^p,\dots)=\frac{X_i^p}{(1-D^p)(1-E^p)+\frac{1}{2}D^p(1-E^p)+\frac{1}{2}(1-D^p)E^p+\frac{1}{3}D^pE^p}.
\end{equation}

The resulting vector
\[(d,e,x_1,\dots x_n)=\Psi_i(D^p,E^p,X_1^p,\dots,X_n^p)\]
corresponds to the correct a-prior probabilities.
\end{theorem}
\begin{proof}
For the first part of the Theorem we directly apply Corollary \ref{aprior_aposterior_n2}. 
For the second part we investigate the a-posterior probability $X^p$ of a third event under the influence of death and emigration.

The event takes place
\begin{itemize}
    \item  with probability $X^p$, in case no death and no emigration occurs,
    \item  with probability $X^p/2$, in case death but no emigration is triggered (in 1/2 of the cases, the event is scheduled earlier than the death event),
    \item  with probability $X^p/2$, in case emigration but no death is triggered (in 1/2 of the cases, the event is scheduled earlier than the death event),
    \item  with probability $X^p/3$, in case death and emigration are triggered (in 1/3 of the cases, the $X$-event is scheduled earlier than the other two).
\end{itemize}
Summing up the cases with the corresponding probabilities leads
\[X^p=X^p\left((1-D^p)(1-E^p)+\frac{1}{2}(1-D^p)E^p+\frac{1}{2}D^p(1-D^p)+\frac{1}{3}D^pE^p\right).\]
The expression in the parenthesis is precisely the stated linear factor $C$ which can be divided to the left-hand side.
\end{proof}
If the data are flawed and $E^p+D^p>1$ - meaning the chance of either emigrating or dying is greater than 1.0 - then we need to follow an alternative approach, since the expression under the square root could become negative. We define
\begin{equation}
    D^p=\left\lbrace\begin{array}{cc}1.0,&D^p\geq E^p,\\
    \frac{2D^p}{D^p+E^p},&D^p<E^p
    \end{array}\right.,\quad E^p=\left\lbrace\begin{array}{cc}\frac{2E^p}{D^p+E^p},&D^p\geq E^p,\\
    1.0,&D^p<E^p.
    \end{array}\right.
\end{equation}
The choice is reasoned by the idea that we (a) want to guarantee that one of the events happen and (b) want to conserve the ratio between the observed probabilities. Let, without loss of generality, $D^p\geq E^p$, then we set $D^p=1$. So death will always trigger, if it is scheduled earlier than emigration. Let $x$ stand for the unknown a-prior probability for emigration, we find
\[P(\text{death})=P(\neg\text{emigration})+P(\text{d scheduled earlier than e})P(\text{emigration})=(1-x)+\frac{1}{2}x=1-\frac{x}{2}.\]
Analogously,
\[P(\text{emigration})=P(\text{e scheduled earlier than d})P(\text{emigration}) =\frac{1}{2}x.\]
It remains to solve
\[\frac{D^p}{E^p}\overset{!}{=}\frac{P(\text{death})}{P(\text{emigration})}=\frac{1-\frac{x}{2}}{\frac{x}{2}}=\frac{2}{x}-1.\]
which leads
\[E^p=x=\frac{2E^p}{D^p+E^p}.\]
The equation for $D^p$ follows analogous. Note that we do not need to change anything for the compensation factor $C$, besides clamping
\[X^p=\min(X^p/C,1.0).\]


\bibliographystyle{plainurl}  
\bibliography{References}

\appendix 
\section{Appendix}
\subsection{Synthetic Internal-Migration Mini-Case-Study}
\label{sec:minimodel}
In the following we introduce a tiny synthetic country to describe the ideas behind the Biregional, Interregional and Full Regional internal migration models introduced in Section \ref{sec:immodel}. To avoid problems due to stochasticity we establish simple deterministic and macroscopic mean-field analogues to the three models which behave like the microscopic versions on the mean value. Moreover, we neglect that individuals become older in the course of a year to allow computation of probabilities by simple divisions.

The study setup is defined as follows:
\begin{itemize}
    \item We define a fictional population of the country and define how many persons internally migrate between the different regions within a given year. This synthetic census will furthermore pose as a the ground truth. 
    \item Dependent on the model, different aspects of the ground truth will be known. Note that the Full Regional model parametrised with the perfectly known census will be able to fully reproduce it.
    \item In the next steps we compute the age-dependent internal emigration probability from the destination-aggregated census, the origin-destination flows from the age-aggregated census, and the internal immigration probabilities from the origin-aggregated census.
    \item We furthermore use these probabilities to evaluate the simulated internal migrants with the Biregional and the Interregional model and compare the outcomes with the census.
    \item We finally investigate ideas to combine the origin-destination flows and the internal immigration probabilities to a feasible parametrisation of the Full Regional model even without perfect knowledge of the census. We run the model and compare the outcomes with the census.
\end{itemize}

\subsubsection{Synthetic Census}
Our synthetic country is defined with three regions $A,B$ and $C$. The inhabitants are either $1$ or $2$ years old and we do not differentiate between sex.

We furthermore assume the following population
\begin{table}[H]
\centering
\begin{tabular}{|c|ccc|c|}
\hline
\multicolumn{5}{|c|}{Synthetic Census: Population}\\
    \hline
     \backslashbox{age}{region} & A & B & C & A+B+C\\
     \hline
     1 & $100$ & $200$ & $100$ & $400$\\
     2 & $200$ & $200$ & $100$ & $500$\\
     \hline
     1+2 & $300$ & $400$ & $200$ & $900$\\
     \hline
\end{tabular}
\caption{Synthetic Census: Population}
\label{tbl:synt_census_pop}
\end{table}

and the following internal migrations within the regarded year:

\begin{table}[H]
\centering
\begin{tabular}{|c|ccc|c|ccc|c|ccc|c|}
    \hline
    \multicolumn{13}{|c|}{Synthetic Census: Internal Migrants}\\
    \hline
    age & \multicolumn{4}{|c|}{1} & \multicolumn{4}{|c|}{2} & \multicolumn{4}{|c|}{1+2}\\
    \hline
     \backslashbox{from}{to} & A & B & C & A+B+C & A & B & C & A+B+C & A & B & C & A+B+C\\
     \hline
     A & $1$ & $5$ & $1$ & $7$& $2$ & $2$ & $10$& $14$ & $3$ & $7$ & $11$ & $21$\\
     B & $2$ & $2$ & $10$ & $14$& $10$ & $2$ & $2$& $14$ &$12$ & $4$ & $12$ & $28$\\
     C & $5$ & $1$ & $1$ & $7$ & $1$ & $5$ & $1$& $7$ & $6$ & $6$ & $2$ & $14$\\
     \hline
     A+B+C & $8$ & $8$ & $12$ & $28$ & $13$ & $9$ & $13$ & $35$ & $21$ & $17$ & $25$& $63$ \\
     \hline
\end{tabular}
\caption{Synthetic Census: Internal Migrants}
\label{tbl:synt_census_im}
\end{table}

We call these two tables \textit{synthetic census} and use them to compute probabilities required for modelling.

\subsubsection{Internal Emigration}
Dividing the number of emigrants per age and origin region (rows A, B, C and columns 1/A+B+C, 2/A+B+C in Table \ref{tbl:synt_census_im}) by the corresponding population (rows 1, 2 and columns A, B, C in Table \ref{tbl:synt_census_pop}), we get the following age-dependent emigration probabilities:

\begin{table}[H]
\centering
\begin{tabular}{|c|ccc|}
\hline
\multicolumn{4}{|c|}{Internal Emigration Probablity $IE^p$}\\
\hline
\backslashbox{age}{region}&A&B&C\\
\hline
1 & $7\%$ & $7\%$ & $7\%$\\
2 & $7\%$ & $7\%$ & $7\%$\\
\hline
\end{tabular}
\end{table}

We see, that the emigration probability is actually age and region independent. That means, all individuals have equal chance to emigrate. This was a deliberate choice in the study design, since internal emigration is not the (most) interesting process when comparing the three models.

\subsubsection{Interregional Model}

Dividing the nine values for the age aggregated migration census (rows A, B, C and columns 1+2/A, 1+2/B, 1+2/C in Table \ref{tbl:synt_census_im}) by the corresponding row-sum (rows A, B, C and column 1+2/A+B+C in Table \ref{tbl:synt_census_im}) we get the age-independent origin-destination probabilities for the Interregional model.
\begin{table}[H]
\centering
\begin{tabular}{|c|ccc|c|}
\hline
\multicolumn{5}{|c|}{Origin-Destination Probabilites $OD^p$}\\
\hline
\backslashbox{from}{to}&A&B&C&$\text{A}\vee\text{B}\vee\text{C}$\\
\hline
A & $14.286\%$ & $33.333\%$ & $52.381\%$& $100\%$\\
B & $42.857\%$ & $14.286\%$ & $42.857\%$& $100\%$\\
C & $42.857\%$ & $42.857\%$ & $14.286\%$& $100\%$\\
\hline
\end{tabular}
\end{table}

With $M^{ir}(i,a,j)=P(i,a)\dot IE^p(i,a)\cdot OD^p(i,j)$ we get the modelled internal migrants. Numbers matching the fictional census in Table \ref{tbl:synt_census_im} are written bold.

\begin{table}[H]
\centering
\begin{tabular}{|c|ccc|c|ccc|c|ccc|c|}
    \hline
    \multicolumn{13}{|c|}{Interregional Model: Modelled Internal Migrants}\\
    \hline
    age & \multicolumn{4}{|c|}{1} & \multicolumn{4}{|c|}{2} & \multicolumn{4}{|c|}{1+2}\\
    \hline
     \backslashbox{from}{to} & A & B & C & A+B+C & A & B & C & A+B+C & A & B & C & A+B+C\\
     \hline
     A & $1$ & $2.333$ & $3.667$ & $\mathbf{7}$ & $2$ & $4.667$ & $7.333$& $\mathbf{14}$ & $\mathbf{3}$ & $\mathbf{7}$ & $\mathbf{11}$ & $\mathbf{21}$\\
     B & $6$ & $2$ & $6$ & $\mathbf{14}$& $6$ & $2$ & $6$& $\mathbf{14}$ &$\mathbf{12}$ & $\mathbf{4}$ & $\mathbf{12}$ & $\mathbf{28}$\\
     C & $3$ & $3$ & $1$ & $\mathbf{7}$ & $3$ & $3$ & $1$& $\mathbf{7}$ & $\mathbf{6}$ & $\mathbf{6}$ & $\mathbf{2}$ & $\mathbf{14}$\\
     \hline
     A+B+C & $10$ & $7.333$ & $10.667$ & $\mathbf{28}$ & $11$ & $9.667$ & $14.333$ & $\mathbf{35}$ & $\mathbf{21}$ & $\mathbf{17}$ & $\mathbf{25}$& $\mathbf{63}$ \\
     \hline
\end{tabular}
\end{table}
We see, that the model outcome matches the fictional census for the age sums (1+2 columns) and for the region sums (A+B+C columns). The prior is explained by how the OD probabilities were gathered, the latter is explained by the age-dependent emigration probabilities.

\subsubsection{Biregional Model}

Dividing the origin-aggregated values for the three destination regions and the two age classes (row A+B+C and columns 1/A, 1/B, 1/C, 2/A, 2/B, 2/C in Table \ref{tbl:synt_census_im}) by the corresponding row-sum (row A+B+C and columns 1/A+B+C, 2/A+B+C in Table \ref{tbl:synt_census_im}) we get the age-dependent internal immigration probabilities $II$ for the Biregional model.
\begin{table}[H]
\centering
\begin{tabular}{|c|ccc|c|}
\hline
\multicolumn{5}{|c|}{Internal Immigration Probabilities $II^p$}\\
\hline
\backslashbox{age}{region}&A&B&C&$\text{A}\vee\text{B}\vee\text{C}$\\
\hline
1 & $28.571\%$ & $28.571\%$ & $42.858\%$ & $100\%$\\
2 & $37.143\%$ & $25.714\%$ & $37.143\%$ & $100\%$\\
\hline
\end{tabular}
\end{table}

With $M^{br}(i,a,j)=P(i,a)\cdot IE^p(i,a)\cdot II^p(j,a)$ we get the modelled internal migrants. Numbers matching the fictional census from Table \ref{tbl:synt_census_im} are written bold.

\begin{table}[H]
\centering
\begin{tabular}{|c|ccc|c|ccc|c|ccc|c|}
    \hline
    \multicolumn{13}{|c|}{Biregional Model: Modelled Internal Migrants}\\
    \hline
    age & \multicolumn{4}{|c|}{1} & \multicolumn{4}{|c|}{2} & \multicolumn{4}{|c|}{1+2}\\
    \hline
     \backslashbox{from}{to} & A & B & C & A+B+C & A & B & C & A+B+C & A & B & C & A+B+C\\
     \hline
     A & $2$ & $2$ & $3$ & $\mathbf{7}$& $5.2$ & $3.6$ & $5.2$& $\mathbf{14}$ & $7.2$ & $5.6$ & $8.2$ & $\mathbf{21}$\\
     B & $4$ & $4$ & $6$ & $\mathbf{14}$& $5.2$ & $3.6$ & $5.2$& $\mathbf{14}$ &$9.2$ & $7.6$ & $11.6$ & $\mathbf{28}$\\
     C & $2$ & $2$ & $3$ & $\mathbf{7}$ & $2.6$ & $1.8$ & $2.6$& $\mathbf{7}$ & $4.6$ & $3.8$ & $5.6$ & $\mathbf{14}$\\
     \hline
     A+B+C & $\mathbf{8}$ & $\mathbf{8}$ & $\mathbf{12}$ & $\mathbf{28}$ & $\mathbf{13}$ & $\mathbf{9}$ & $\mathbf{13}$ & $\mathbf{35}$ & $\mathbf{21}$ & $\mathbf{17}$ & $\mathbf{25}$& $\mathbf{63}$ \\
     \hline
\end{tabular}
\end{table}

Like the interregional model, the results match the census for the overall emigrants (A+B+C columns), which is a consequence of using the same internal emigration model. Compared to the interregional model, the results match for the overall age structure of the immigrated agents (A+B+C row), but the validity of the 1+2 columns, i.e. the overall flows, is lost.

\subsubsection{Full Regional Model}
Finally, we may compute the probabilities for the Full Regional model $II_2$ by dividing the individual data cells in Table \ref{tbl:synt_census_im} by their row sum (A+B+C columns).

\begin{table}[H]
\centering
\begin{tabular}{|c|ccc|c|ccc|c|}
    \hline
    \multicolumn{9}{|c|}{Internal Migration Probabilities $II_2^p$}\\
    \hline
    age & \multicolumn{4}{|c|}{1} & \multicolumn{4}{|c|}{2} \\
    \hline
     \backslashbox{from}{to} & A & B & C & A+B+C & A & B & C & A+B+C \\
     \hline
     A & $14.286\%$ & $71.428\%$ & $14.286\%$ & $100\%$& $14.286\%$ & $14.286\%$ & $71.428\%$& $100\%$\\
     B & $14.286\%$& $14.286\%$ & $71.428\%$ & $100\%$& $71.428\%$ & $14.286\%$ & $14.286\%$& $100\%$\\
     C & $71.428\%$ & $14.286\%$ & $14.286\%$ & $100\%$ & $14.286\%$ & $71.428\%$ & $14.286\%$& $100\%$\\
     \hline
\end{tabular}
\end{table}

With these probabilities, finally, the model results with $M(i,a,j)=P(i,a)\cdot IE^p(i,a)\cdot II_2^p(i,a,j)$ are identical with the synthetic census from Table \ref{tbl:synt_census_im}.

\begin{table}[H]
\centering
\begin{tabular}{|c|ccc|c|ccc|c|ccc|c|}
    \hline
    \multicolumn{13}{|c|}{Full Regional Model: Modelled Internal Migrants}\\
    \hline
    age & \multicolumn{4}{|c|}{1} & \multicolumn{4}{|c|}{2} & \multicolumn{4}{|c|}{1+2}\\
    \hline
     \backslashbox{from}{to} & A & B & C & A+B+C & A & B & C & A+B+C & A & B & C & A+B+C\\
     \hline
     A & $\mathbf{1}$ & $\mathbf{5}$ & $\mathbf{1}$ & $\mathbf{7}$& $\mathbf{2}$ & $\mathbf{2}$ & $\mathbf{10}$& $\mathbf{14}$ & $\mathbf{3}$ & $\mathbf{7}$ & $\mathbf{11}$ & $\mathbf{21}$\\
     B & $\mathbf{2}$ & $\mathbf{2}$ & $\mathbf{10}$ & $\mathbf{14}$& $\mathbf{10}$ & $\mathbf{2}$ & $\mathbf{2}$& $\mathbf{14}$ &$\mathbf{12}$ & $\mathbf{4}$ & $\mathbf{12}$ & $\mathbf{28}$\\
     C & $\mathbf{5}$ & $\mathbf{1}$ & $\mathbf{1}$ & $\mathbf{7}$ & $\mathbf{1}$ & $\mathbf{5}$ & $\mathbf{1}$& $\mathbf{7}$ & $\mathbf{6}$ & $\mathbf{6}$ & $\mathbf{2}$ & $\mathbf{14}$\\
     \hline
     A+B+C & $\mathbf{8}$ & $\mathbf{8}$ & $\mathbf{12}$ & $\mathbf{28}$ & $\mathbf{13}$ & $\mathbf{9}$ & $\mathbf{13}$ & $\mathbf{35}$ & $\mathbf{21}$ & $\mathbf{17}$ & $\mathbf{25}$& $\mathbf{63}$ \\
     \hline
\end{tabular}
\end{table}

\subsubsection{Model Comparison}
One of the key questions of this model comparison is, whether the probability $II$ from the Biregional model can somehow be combined with the probability $OD$ from the Interregional model, to be valid in both ``worlds'': the age-distribution of the immigrants (row A+B+C) and the overall flows (columns 1+2). This way, we could generate a well working migration model without knowing the full synthetic census or even the probability table $II_2^p$. 

Let $p(i,a,j)$ denote the probabilities of interest, then constraints would be written as:
\[
    \forall a\in \{1,2\},j\in\{A,B,C\}:\sum_{i\in \{A,B,C\}}P(i,a)\cdot IE(i,a)^p\cdot p(i,a,j)=\sum_{i\in \{A,B,C\}}P(i,a)\cdot IE^p(i,a)\cdot II^p(j,a),
\]
\[
    \forall i,j\in \{A,B,C\}:\sum_{a\in \{1,2\}}P(i,a)\cdot IE^p(i,a)\cdot p(i,a,j)=\sum_{a\in \{1,2\}}P(i,a)\cdot IE^p(i,a)\cdot OD^p(i,j).
\]
\[
    \forall a\in \{1,2\},i\in \{A,B,C\}:\sum_{j\in \{A,B,C\}}p(i,a,j)=1.
\]
We define $IE(i,a):=P(i,a)\cdot IE^p(i,a)$ and transform the equations to get a better picture:
\begin{equation}
    \forall j\in\{A,B,C\},a\in \{1,2\}:\sum_{i\in \{A,B,C\}}\frac{IE(i,a)}{\sum_{k\in \{A,B,C\}} IE(k,a)}\cdot p(i,a,j)=II^p(j,a),
\end{equation}
\begin{equation}
    \forall i,j\in \{A,B,C\}:\sum_{a\in \{1,2\}}\frac{IE(i,a)}{\sum_{b\in \{1,2\}} IE(i,b)}\cdot p(i,a,j)=OD^p(i,j).
\end{equation}
\begin{equation}
    \forall a\in \{1,2\},i\in \{A,B,C\}:\sum_{j\in \{A,B,C\}}p(i,a,j)=1.
\end{equation}
Additional constraint
\begin{equation}
    \forall a\in \{1,2\},i,j\in \{A,B,C\}:0\leq p(i,a,j)\leq 1
\end{equation}
must be met to justify the use of $p$ as probability. This equation systems seems solvable in form of a linear program: In this particular case we have $2\cdot3+3\cdot 3+2\cdot 3=21$ constraint equations and $3\cdot 3 \cdot 2=18$ degrees of freedom, in the typical case, i.e. with more age classes and regions, we usually receive by far more degrees of freedom than constraint equations.

There are various ways to tackle this problem, e.g using existing libraries for linear programming. Considering that the problem has $q_{im}\cdot (a_{max}+1)\cdot q_{im}$ free variables with $2q_{im}\cdot (a_{max}+1)+q_{im}^2$ constraint equations, which both might easily lie in the Millions for reasonable number of age classes and regions, an exact solution might become difficult and tailored heuristic approaches might be more suitable.

Below we see a Table with one solution for $p$, which was found with a tailored metaheuristic. The and corresponding model results are seen below. It is not by accident that the results are whole numbers, since the metaheuristic works with absolute numbers instead of probabilities.

\begin{table}[H]
\centering
\begin{tabular}{|c|ccc|c|ccc|c|}
    \hline
    \multicolumn{9}{|c|}{Estimated Internal Migration Probabilities $p$ from $II^p$ and $OD^p$}\\
    \hline
    age & \multicolumn{4}{|c|}{1} & \multicolumn{4}{|c|}{2} \\
    \hline
     \backslashbox{from}{to} & A & B & C & A+B+C & A & B & C & A+B+C \\
     \hline
     A & $14.286\%$ & $28.571\%$ & $57.143\%$ & $100\%$& $14.286\%$ & $35.714\%$ & $50\%$& $100\%$\\
     B & $35.714\%$& $14.286\%$ & $50\%$ & $100\%$& $50\%$ & $14.286\%$ & $35.714\%$& $100\%$\\
     C & $28.571\%$ & $57.143\%$ & $14.286\%$ & $100\%$ & $57.143\%$ & $28.571\%$ & $14.286\%$& $100\%$\\
     \hline
\end{tabular}
\end{table}

\begin{table}[H]
\centering
\begin{tabular}{|c|ccc|c|ccc|c|ccc|c|}
    \hline
    \multicolumn{13}{|c|}{Full Regional Model using $p$ as $II^2$: Modelled Internal Migrants}\\
    \hline
    age & \multicolumn{4}{|c|}{1} & \multicolumn{4}{|c|}{2} & \multicolumn{4}{|c|}{1+2}\\
    \hline
     \backslashbox{from}{to} & A & B & C & A+B+C & A & B & C & A+B+C & A & B & C & A+B+C\\
     \hline
     A & $1$ & $2$ & $4$ & $\mathbf{7}$ & $2$ & $5$ & $7$& $\mathbf{14}$ & $\mathbf{3}$ & $\mathbf{7}$ & $\mathbf{11}$ & $\mathbf{21}$\\
     B & $5$ & $2$ & $7$ & $\mathbf{14}$& $7$ & $2$ & $5$& $\mathbf{14}$ &$\mathbf{12}$ & $\mathbf{4}$ & $\mathbf{12}$ & $\mathbf{28}$\\
     C & $2$ & $4$ & $1$ & $\mathbf{7}$ & $4$ & $2$ & $1$& $\mathbf{7}$ & $\mathbf{6}$ & $\mathbf{6}$ & $\mathbf{2}$ & $\mathbf{14}$\\
     \hline
     A+B+C & $\mathbf{8}$ & $\mathbf{8}$ & $\mathbf{12}$ & $\mathbf{28}$ & $\mathbf{13}$ & $\mathbf{9}$ & $\mathbf{13}$ & $\mathbf{35}$ & $\mathbf{21}$ & $\mathbf{17}$ & $\mathbf{25}$& $\mathbf{63}$ \\
     \hline
\end{tabular}
\end{table}
Compared to Table \ref{tbl:synt_census_im}, the model-result fulfils the two required balance equations (all flows are correct for the 1+2 column, and the A+B+C row is correct for both ages 1 and 2), yet the age dependent flows between the individual regions still don't have very much in common with the original synthetic census, or the Full Regional model with the correct parameters respectively. Apparently, even for this minimalist example the problem is actually under-determined and multiple (infinite) solutions exist. Although this is intuitively quite clear, it is mathematically surprising: Due to the number of equations ($21$ equations for $18$ degrees of freedom) the task seemed over-determined the first glance.

\end{document}